\documentclass[english,journal]{IEEEtran}

\usepackage[style=ieee, doi=false, isbn=false, url=false]{biblatex}
\usepackage{hyperref}
\usepackage{bm,bbm}
\usepackage{amsthm,amssymb,amsmath,mathtools}
\usepackage{stmaryrd}
\usepackage{microtype}
\usepackage{algpseudocode,algorithm,algorithmicx}
\usepackage[svgnames]{xcolor}
\usepackage{graphicx,tabularx,float}
\usepackage{babel,csquotes}
\usepackage{comment}
\usepackage[normalem]{ulem}

\makeatletter

\theoremstyle{plain}
\newtheorem{theorem}{Theorem}
\newtheorem{definition}[theorem]{Definition}
\newtheorem{lemma}[theorem]{Lemma}
\newtheorem{proposition}[theorem]{Proposition}
\newtheorem{corollary}[theorem]{Corollary}

\bmdefine{\bA}{A}
\bmdefine{\ba}{a}
\bmdefine{\bB}{B}
\bmdefine{\bb}{b}
\bmdefine{\bC}{C}
\bmdefine{\bc}{c}
\bmdefine{\bD}{D}
\bmdefine{\bd}{d}
\bmdefine{\bE}{E}
\bmdefine{\be}{e}
\bmdefine{\bF}{F}
\bmdefine{\bf}{f}
\bmdefine{\bG}{G}
\bmdefine{\bg}{g}
\bmdefine{\bH}{H}
\bmdefine{\bh}{h}
\bmdefine{\bI}{I}
\bmdefine{\bi}{i}
\bmdefine{\bJ}{J}
\bmdefine{\bj}{j}
\bmdefine{\bK}{K}
\bmdefine{\bk}{k}
\bmdefine{\bL}{L}
\bmdefine{\bl}{l}
\bmdefine{\bM}{M}
\bmdefine{\bmm}{m}
\bmdefine{\bN}{N}
\bmdefine{\bn}{n}
\bmdefine{\bO}{O}
\bmdefine{\bo}{o}
\bmdefine{\bP}{P}
\bmdefine{\bp}{p}
\bmdefine{\bQ}{Q}
\bmdefine{\bq}{q}
\bmdefine{\bR}{R}
\bmdefine{\br}{r}
\bmdefine{\bS}{S}
\bmdefine{\bs}{s}
\bmdefine{\bT}{T}
\bmdefine{\bt}{t}
\bmdefine{\bU}{U}
\bmdefine{\bu}{u}
\bmdefine{\bV}{V}
\bmdefine{\bv}{v}
\bmdefine{\bW}{W}
\bmdefine{\bw}{w}
\bmdefine{\bX}{X}
\bmdefine{\bx}{x}
\bmdefine{\bY}{Y}
\bmdefine{\by}{y}
\bmdefine{\bZ}{Z}
\bmdefine{\bz}{z}

\bmdefine{\balpha}{\alpha}
\bmdefine{\bbeta}{\beta}
\bmdefine{\bgamma}{\gamma}
\bmdefine{\bGamma}{\Gamma}
\bmdefine{\bdelta}{\delta}
\bmdefine{\bDelta}{\Delta}
\bmdefine{\bepsilon}{\epsilon}
\bmdefine{\bvarepsilon}{\varepsilon}
\bmdefine{\bzeta}{\zeta}
\bmdefine{\bmeta}{\eta}
\bmdefine{\btheta}{\theta}
\bmdefine{\bTheta}{\Theta}
\bmdefine{\biota}{\iota}
\bmdefine{\bkappa}{\kappa}
\bmdefine{\blambda}{\lambda}
\bmdefine{\bLambda}{\Lambda}
\bmdefine{\bmu}{\mu}
\bmdefine{\bnu}{\nu}
\bmdefine{\bpi}{\pi}
\bmdefine{\bPi}{\Pi}
\bmdefine{\brho}{\rho}
\bmdefine{\bsigma}{\sigma}
\bmdefine{\bSigma}{\Sigma}
\bmdefine{\btau}{\tau}
\bmdefine{\bupsilon}{\upsilon}
\bmdefine{\bUpsilon}{\Upsilon}
\bmdefine{\bphi}{\phi}
\bmdefine{\bPhi}{\Phi}
\bmdefine{\bchi}{\chi}
\bmdefine{\bpsi}{\psi}
\bmdefine{\bPsi}{\Psi}
\bmdefine{\bomega}{\omega}
\bmdefine{\bOmage}{\Omega}

\newcommand{\cB}{{\mathcal{B}}}

\newcommand{\cN}{\mathcal{N}}

\newcommand{\bbP}{\mathbb{P}}
\newcommand{\bbR}{\mathbb{R}}

\newcommand{\sT}{{\mathsf{T}}}
\newcommand{\sF}{{\mathsf{F}}}

\newcommand{\rE}{\mathrm{E}}

\DeclareMathOperator*{\argmin}{arg\,min}

\DeclareMathOperator{\cov}{Cov}
\DeclareMathOperator{\diag}{diag}

\DeclareMathOperator{\sign}{sgn}

\DeclareMathOperator{\KL}{KL}

\let\hat\widehat
\let\tilde\widetilde

\newcommand\inp[2]{\left\langle #1,#2 \right\rangle}
\newcommand\norm[1]{\left\lVert #1 \right\rVert}

\newcommand{\ellp}{{\ell^\prime}}

\makeatother

\usepackage{soul}
\usepackage{enumitem}

\hypersetup{
    pdftitle={Guaranteed Private Communication with Secret Block Structure},
    pdfauthor={Maxime Ferreira Da Costa, Jianxiu Li, and Urbashi Mitra}
}

\newcommand*\Let[2]{\State #1 $\gets$ #2}

\title{Guaranteed Private Communication \\ with Secret Block Structure}

\author{
    Maxime Ferreira Da Costa,
	 Jianxiu Li,
  and Urbashi Mitra
	 \thanks{M. Ferreira Da Costa is with the Laboratory of Signals and Systems (L2S) at CentraleSupélec, Université Paris--Saclay. The work of M. Ferreira~Da~Costa is supported in part by ANR-20-IDEES-0002. Email: \texttt{maxime.ferreira@centralesupelec.fr}. }
	 \thanks{J. Li and U. Mitra are with the Viterbi School of Engineering, University of Southern California. The work of J. Li and U. Mitra is supported in part by the USC + Amazon Center on Secure and Trusted Machine Learning, NSF CCF-1817200, DOE DE-SC0021417, Swedish Research Council 2018-04359, NSF CCF-2008927, NSF CCF-2200221, ONR 503400-78050, and ONR N00014-15-1-2550. Emails: \texttt{jianxiul@usc.edu}, \texttt{ubli@usc.edu.}}
}

\allowdisplaybreaks

\begin{document}

\maketitle

\begin{abstract}
A novel private communication framework is proposed where privacy is induced by transmitting over a channel instances of linear inverse problems that are identifiable to the legitimate receiver but unidentifiable to an eavesdropper. The gap in identifiability is created in the framework by leveraging secret knowledge between the transmitter and the legitimate receiver. Specifically, the case where the legitimate receiver harnesses a secret block structure to decode a transmitted block-sparse message from underdetermined linear measurements in conditions where classical compressed sensing would provably fail is examined. The applicability of the proposed scheme to practical multiple-access wireless communication systems is discussed. The protocol's privacy is studied under a single transmission, and under multiple transmissions without refreshing the secret block structure. It is shown that, under a specific scaling of the channel dimensions and transmission parameters, the eavesdropper can attempt to overhear the block structure from the fourth-order moments of the channel output. Computation of a statistical lower bound suggests that the proposed fourth-order moment secret block estimation strategy is near optimal. The performance of a spectral clustering algorithm is studied to that end, defining scaling laws on the lifespan of the secret key before the communication is compromised. Finally, numerical experiments corroborating the theoretical findings are conducted.
\end{abstract}


\section{Introduction}

\IEEEPARstart{W}{hile}  communication privacy is often ensured at higher network layers~\cite{Yu, Tomasin, Schmitt}, and can be achieved via cryptographic means; there are new methods in
\emph{physical layer security}~\cite{bloch2011physical}, which can leverage the structural properties of a communication channel to generate privacy. Physical layer privacy can
strengthen security in modern data exchange protocols, such as next-generation wireless systems, the Internet of Things, and satellite constellations. Physical layer security offers numerous complementary guarantees to usual cryptography: It can protect users' identities, physical locations, or even conceal the existence of a communication to an eavesdropper; and can be implemented opportunistically over wireless channels with no or little computational overhead.  There is interest in realizing the theoretical promises of physical layer security in realistic systems~\cite{poor2017wireless}.

Traditional physical layer privacy schemes exploit channel differences to share information with Bob without Eve's knowledge, which often comes with the assumption that Bob and Eve's channels are distinct. Typical strategies involve the use of artificial noise~\cite{goel2008GuaranteeingSecrecy,tomasin2022BeamformingArtificial,rajiv2022securing,krunz2023secure}. The noise can be either injected into the nullspace of channel state information (CSI) and mitigated by exploiting CSI or directly injected noise into the transmitted message and resolved by the legitimate receiver side by exploiting a secret key~\cite{zhang2018CovertCommunication,schaefer2018SecureBroadcasting}. Other privacy schemes involve random and adversarial beamforming design~\cite{ayyalasomayajula2023users,Checa}, or the injection of fake paths over geometric channels to diminish the capability of an eavesdropper to distinguish between true and fake paths and challenge the estimation of CSI~\cite{li2023ChannelState,tran2024physical} by an eavesdropper.

The previously mentioned physical layer security schemes induce privacy by performing a linear action on the transmitted message that is statistically hard to invert without additional knowledge. In a related fashion, the compressed sensing framework~\cite{donoho2006compressed} assumes a non-linear prior on the input message and has been exploited to ensure privacy~\cite{zhang2016review}. If the sensing matrix is kept secret to an eavesdropper, perfect secrecy can be guaranteed in the information-theoretic sense~\cite{liang2009information} under restrictive conditions~\cite{bianchi2015analysis}. Typical sensing matrices are functions of the CSI. The computational secrecy of this approach has also been investigated~\cite{orsdemir2008security,rachlin2008secrecy}, restricting Eve's ability to recover the encoded message via a polynomial time algorithm.

Motivated by applications to multiple access wireless systems, we focus here, instead, on a novel model where the sensing matrix (\emph{e.g.} the channel matrix) is imposed by the environment and is \emph{not} under the control of the transmitter. Privacy is achieved by sharing an additional structure with the legitimate receiver, easing the decoding of the message~\cite{baraniuk2010model}. From the eavesdropper's perspective, the decoding amounts to solving a bilinear inverse problem, which is known to demand much more stringent assumptions to be identifiable~\cite{choudhary2018properties,choudhary2013identifiability2,da2019self,li2016identifiability,lee2018fast,ahmed2013blind}. Thus, statistical hardness is exploited to provide privacy.

\subsection{Linear Inverse Problem Based Privacy}
We consider the classical secret communication problem with side information: A transmitter (Alice) wishes to privately transmit a vector $\bx \in \bbR^n$ to a legitimate receiver (Bob) over a public channel\footnote{For clarity purposes, signals are assumed to be real-valued, but the present model and analysis are extendable to the complex case.}. The noisy channel output received by Bob and the eavesdropper (Eve) are $\by_B = f_B(\bx) + \bw_B$ and $\by_E = f_E(\bx) + \bw_E$, with noise $\bw_B$ and $\bw_E$, respectively. To achieve privacy and prevent Eve from recovering the message $\bx$, Alice and Bob may communicate a \emph{low information rate signal} over a secure channel inaccessible to Eve. The secure channel assumption is common in physical layer security and has been previously used
with collaborative inference strategies~\cite{mohapatra2016CapacityTwoUser}, or in covert communication~\cite{zhang2021CovertCommunication}. This secure channel can be constructed, for example, through coding such as in the context of a wiretap channel with side information~\cite{oggier2011SecrecyCapacity,chen2008wiretap}.

\begin{figure}[t]
	\centering	\includegraphics[width=\columnwidth]{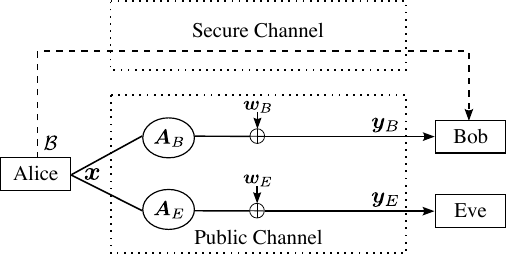}
	\caption{Communication model with secure channel.} \label{fig:model}
\end{figure}

In the proposed setting, the effect of the Alice--Bob and Alice--Eve channels are assumed to be linear and modeled by a ``fat'' matrices $\bA_{B} \in \bbR^{m \times n}$, and $\bA_{E} \in \bbR^{m \times n}$, respectively, with $m < n$ so that channel outputs $\by_B$ and $\by_E$ write
\begin{subequations}
    \begin{align}
    \by_B &= \bA_B \bx + \bw_B \\
    \by_E &= \bA_E \bx + \bw_E
\end{align}
\end{subequations}
where $\bw_B \sim \cN(\bm{0},\sigma^2_B \bI_{m})$ and $\bw_E \sim \cN(\bm{0},\sigma^2_E \bI_{m})$ are white Gaussian noise. The matrices $\bA_B$ and $\bA_E$ are imposed by the \textit{environment}; $\bA_B$ is known by Bob and $\bA_E$ is known by Eve. Finally, we assume that Eve is aware of the communication protocol established by Alice. The overall communication model is depicted in Figure~\ref{fig:model}.

For the purposes of our analysis, we will assume the channel matrices $\bA_B$ and $\bA_E$ to satisfy certain incoherence properties, which are detailed in the sequel. The privacy results in~Section~\ref{sec:eavesdroppingHighOrderMoments} are given in terms of incoherence and hold \emph{regardless of the specific realizations} of the Alice--Bob and Alice--Eve channels. For this reason, and to improve clarity, the subscripts ``B'' and ``E'' referring to Bob and Eve's model parameters are dropped in the rest of the paper unless a disambiguation is explicitly needed.

To ensure privacy, Alice, who designs the message $\bx$ and the side information, must ensure two properties. First, Bob must be able to provably recover $\bx$ from the observation $\by$ via the side information from the secure channel. Second, Eve cannot provably recover $\bx$ without knowing the side information. Thus, Alice is left to design an inverse problem that is \emph{identifiable} to Bob but \emph{unidentifiable} to Eve. These goals can typically be jointly achieved by imposing an additional structure on $\bx$ and privately sharing this structure over the secure channel. For practicality, this structure must be comprised of a small number of bits and reusable over multiple transmissions. Building onto our prior work~\cite{dacosta2022FrameworkPrivate}, we propose that Alice shares a secret block structure with Bob and encodes her message as a block-sparse signal whose support follows this secret structure. Harnessing a block-sparse prior to recovering signals through underdetermined linear measurements has been extensively shown to allow exact recovery in conditions where classical compressed sensing would provably fail~\cite{eldar2009robust,eldar2009block,baraniuk2010model, gribonval2003sparse,}. We leverage these results to establish the existence of a private communication regime where Alice and Bob achieve secrecy by transmitting single instances of an unidentifiable compressed sensing problem over a public channel. Then, as refreshing the secret block structure at each channel transmission is impractical, we study the privacy of the communication from multiple transmissions while reusing the same secret block structure. We propose a near-optimal method for Eve to eavesdrop on the block structure based on the spectral clustering of the fourth-order moments of the channel output. An upper bound on the number of transmissions before the secret structure and the messages are compromised is derived, and the trade-off between key reuse and secrecy is discussed. Spectral clustering has been considered a fast and robust method to recover low-dimensional structures in high-dimensional datasets with significant success. It has been applied, for instance, to recovering partitions and cliques in high-dimensional graphs~\cite{rohe2011SpectralClustering} as well as for unsupervised classification in machine learning~\cite{von2007tutorial}.

The proposed signaling scheme is motivated by its applicability to modern multi-user wireless communication protocols. As an example, we assume an uplink scenario with $r$ many transmitters sending within a symbol interval a message $\bu_q \in \bbR^{d}$ using a precoding scheme $\bm{S}_q \in \bbR^{n \times d}$ through a linear channel $\bm{H}_q \in \bbR^{n \times n}$ that is imposed by the environment. The received message at the base station classically reads $\by = \sum_{q=1}^r
 \bH_q \bS_q \bu_q  + \bw$. When the channel users parsimoniously transmit at a given symbol interval, that is, a random fraction of users remain inactive, the channel input can be modeled with a group-sparse prior. If this prior is only known by the legitimate base station (Bob), the relative identifiability of block-sparse signals versus unstructured sparse signals can be exploited to induce privacy against an eavesdropper.
 Many massive access communication schemes rely on sporadic channel traffic~\cite{wu2020massiveAccess} to allow more robust decoding on the receiver side, even from an under-determined channel output. We pinpoint two practical schemes where our framework is applicable:
\begin{enumerate}
    \item In \emph{overloaded CDMA communications}, the transmitters rely on unique sequences $\{\bS_1, \cdots, \bS_r \}$, known to the base station, to spread the messages onto a larger dimension space before transmission~\cite{verdu1999spectral,chen2001multicarrier}. Sparse coded multiple access schemes have been considered to improve user detection when the system is overloaded~\cite{alam2018NonOrthogonalMultiple}, and adapted coding sequences are proposed in~\cite{liu2020IdenticalCode,zhu2011ExploitingSparse}. However, the privacy benefits of overloading have not yet been considered in that context.
    \item In \emph{massive MIMO communications}, the number of identifiable spatial streams equals the number of receive antennas. Suppose the transmitter has more antennas than the receiver. In that case, she intermittently activates sub-groups of antennas according to a pattern shared with the receiver. She transmits on the active sub-groups at each symbol interval at the price of a reduced bit rate. Such MIMO systems have been considered to minimize implementation cost~\cite{ni2016HybridBlock} or improve spectral efficacy~\cite{wang2019non,liu2018gaussian}.
\end{enumerate}

\subsection{Contributions and Paper Organization}
We build upon our prior work~\cite{dacosta2022FrameworkPrivate} and present an improved eavesdropping scheme based on fourth moments with complete proofs and numerical simulations. Computation of a statistical lower bound suggests that the improved eavesdropping scheme is asymptotically near-optimal. In Section \ref{sec:secretBlockSparsity}, we propose a novel communication protocol that leverages the advantageous recoverability of block-sparse signals to ensure privacy. We provide Alice and Bob's encoding and decoding strategies, respectively. In our design, Alice transmits secretly to Bob a block structure and uses this structure to encode her message, which can be done at a very low transmission rate, while the channel matrix $\bA$  \emph{cannot} be designed by Alice and is provided by nature. To the authors' knowledge, the proposed protocol is the first linear inverse problem-based privacy method that does not require the matrix $\bA$ to be secretly shared. Additionally, unlike most pre-existing physical layer privacy designs, neither co-location nor distinct locations are needed for the proposed scheme. Furthermore, Corollary \ref{cor:singleSnapshotPrivacy} guarantees that Alice can adjust the block length and the sparsity level of the message she transmits so that the transmission is provably identifiable for Bob and unidentifiable to Eve as the signal length increases.
In Section \ref{sec:eavesdroppingHighOrderMoments}, we consider the possibility of Eve recovering the secret block structure from the observation of {\em multiple} snapshots of the observation $\{\by_\ell\}$ that Alice has generated with the same block structure $\cB$. We show in Proposition~\ref{prop:estimateOfB} that, depending on Alice's choice of the block length and sparsity level, it is possible to extract $\cB$ from the fourth-order moments of the observation and propose an eavesdropping algorithm to that end. We investigate the case of a finite number of snapshots and derive an upper bound on the rate at which Alice must generate a new $\cB$ to prevent Eve from deciphering Bob's messages.

We present numerical results that validate our theoretical findings in Section~\ref{sec:numericalSimulation}. Section~\ref{sec:conclusion} draws a conclusion, and further research directions are discussed.

\subsection{Notations}
Vectors of $\bbR^n$ and matrices of $\bbR^{n_1 \times n_2}$ are denoted by boldface letters $\ba$ and capital boldface letters $\bA$, respectively. The entry $(i,j)$ of a matrix $\bA$ is written as $a_{i,j}$. The matrix norms $\left\Vert \bM \right\Vert_2$, $\norm{\bM}_{\sF}$, and $\left\Vert \bM \right\Vert_{\mathrm{max}}$ refer to the spectral norm, the Frobenius norm, and the maximal absolute value of the entries in $\bM$, respectively. Given a positive semi-definite matrix $\bM$, we write $\lambda_{\min}\left(\bM\right)$ and $\lambda_{r}\left(\bM\right)$ as its smallest eigenvalue, and $r$th-largest eigenvalue (with multiplicity), respectively. The Hadamard product between two matrices $\bM_1$ and $\bM_2$ is denoted as $\bM_1 \odot \bM_2$. We write by $\bI_n$ the identity matrix and by $\bJ_n$ the all-one matrix in dimension $n\times n$. Given a random vector $\bz \in \bbR^{n}$, we denote by $\bSigma_{\bz} \in \bbR^{n \times n}$ its covariance matrix. A block structure over $\bbR^n$ into $r$ blocks is described by a mapping  $\cB : [1,\dots,n] \to [1,\dots,r]$, and is associated with the indicator matrix of $\bB \in \bbR^{n \times n}$ defined by
\begin{equation}\label{eq:definition_blockmatrix}
	b_{i,j} = \begin{cases}
		1 & \textrm{if } \cB(i) = \cB(j) ; \\
		0 & \textrm{if } \cB(i) \neq \cB(j).
		\end{cases}
\end{equation}
We denote by $\bx[q]$ the subvector of $\bx$ with entries $x_i$ ensuring $\cB(i) = q$.  The ``block-$\ell_0$-norm'' of a vector $\bx$ is defined as $\norm{\bx}_{\cB,0} = \sum_{q=1}^r \bm{1}_{\bx[q] \neq \bm{0}}$ and counts the number blocks in $\bx$ that are not exactly equal to $\bm{0}$. For two functions $f$ and $g$, we use the Landau notation $f = o(g)$ to denote that the ratio $\frac{f(t)}{g(t)}$ tends to $0$ as $t \to \infty$.

\section{Privacy with Block Sparsity}\label{sec:secretBlockSparsity}

\subsection{Alice's Encoding}\label{ssec:AliceEncoding}
In the proposed protocol, Alice constructs her message $\bx$ as follows. Given the knowledge of the channel dimension, Alice initializes the communication by randomly selecting a block structure $\cB : [1,\dots,n] \to [1,\dots,r]$. Alice sends this structure to Bob over the secret channel. We highlight that this exchange only requires $n \log_2(r)$ bits of information, which is significantly less than schemes relying on exchanging the entire matrix $\bA \in \mathbb{R}^{m \times n}$ ($mn$ infinite precision numbers). Although not required in practice, we assume for simplicity that the $r$ blocks have equal block size $d$, {\em i.e.} $n = r \cdot d$. Next, Alice selects a probability of block activation $p \in \left[0,\frac{1}{2} \right]$, where $p \leq \frac{1}{2}$ is assumed for convenience in the analysis, and encodes her message in a block-sparse vector $\bx$. In the sequel, we assume that $\bx$ is distributed according to a block Bernouilli--Gaussian distribution such that
\begin{equation}\label{eq:blockBernouilliGaussian}
	\bx[q] = \begin{cases}
		\bm{0}_d   & \textrm{w.p. } 1-p \\
		\bz[q] & \textrm{w.p. } p,
	\end{cases}
\end{equation}
where $\bz[q] \sim \cN(\bm{0}_d, \bI_d) $ is a random i.i.d. standard Gaussian vector of dimension $d$. A visualization of the block sparsity encoding is provided in Figure~\ref{fig:block}.
\begin{figure}[t]
	\centering
    \includegraphics[scale=0.95]{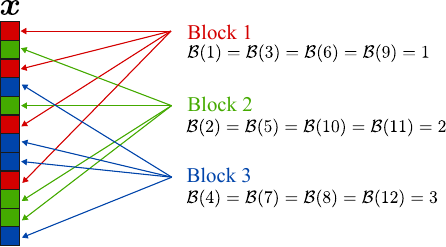}
	\caption{Example of block sparse encoding in dimension $n = 12$, with $r=3$ blocks of length $d=4$. }
	\label{fig:block}
\end{figure}

\subsection{Bob's Decoding} \label{ssec:bobDecoding}
At the public channel output, Bob receives a vector $\by = \bA \bx +\bw$ and leverages $\cB$ that Alice securely sent to recover the ground truth message $\bx$. To do so, Bob formulates the block-compressed sensing problem:
\begin{align}\label{eq:bob_BlockCS}
	\hat{\bx}_{B} = \argmin_{\bx\in\bbR^n} \norm{\bx}_{\cB,0} \textrm{ such that } \left\Vert \by - \bA \bx \right\Vert_2 \leq \epsilon,
\end{align}
where $\epsilon>0$ is a parameter that scales proportionally with the standard deviation of the noise $\left\Vert \bw \right\Vert_2$.
Harnessing a block-sparse prior in compressed sensing has been extensively shown in the literature to enhance the identifiability of \eqref{eq:bob_BlockCS} and to allow an exact reconstruction of the message with much fewer measurements than classical compressed sensing \cite{eldar2009robust, gribonval2003sparse}. However, directly solving \eqref{eq:bob_BlockCS} remains NP-hard in the general case due to the combinatorics inherent to the minimization of $\norm{\bx}_{\cB,0}$. Thus, Bob computes, instead, an estimate of $\widehat{\bx}_{B}$ using a polynomial time algorithm of his choice. Among the many addressed algorithms proposed in the literature, Block Matching Pursuit (Block MP) \cite{bach2008consistency}, Block Iterative Harding Thresholding (Block IHT), Block Basis Pursuit (Block BP) \cite{eldar2010block} or block-based CoSaMP \cite{baraniuk2010model}, have been shown to have provable performance guarantees.

In the sequel, we denote $\beta = \frac{m}{np}$ as the \emph{redundancy parameter}, defined as the ratio between the number of measurements at the channel output and the \emph{expected} number of non-zero entries in the block-sparse input vector $\bx$. We remark that $\beta \geq 1$ is trivially needed to decode the message successfully. Asymptotic phase transitions for the success of greedy algorithms to recover the block-sparse ground truth have been studied in the literature~\cite{baraniuk2010model}. Proposition~\ref{prop:successBob} reinterprets this result in terms of the parameter $\beta$, the block-length $d$, and the transmission parameter $p$ in the asymptotics $n \to \infty$.
\begin{proposition}[Success of Bob's decoding]\label{prop:successBob}
	Suppose that $\bA$ is a matrix with i.i.d. random Gaussian entries and assume a noise-free environment $\bw = \bm{0}$. If
	\begin{equation}\label{eq:BobSuccessRegime}
	    \log\left(\frac{1}{p}\right) = o\left( \frac{d}{\log(d)} \right)
	    \qquad  \text{and}  \qquad  \beta \to \infty
	\end{equation}
	in the limit where $n \to \infty$, then Bob can stably recover $\bx$ asymptotically almost surely.
\end{proposition}
Additionally, denoising bounds on the estimate of the input vector $\bx$ are provided in the presence of noise~\cite{baraniuk2010model}.

\subsection{Privacy Guarantees under a Single Snapshot}\label{ssec:privacySingleSnapshot}

If only one snapshot $\by$ is observed, it is impossible for Eve to reliably infer $\cB$, which remains ambiguous even with perfect knowledge of $\bx$. Therefore, from her perspective, the best possible approach consists of attempting to recover $\bx$ without leveraging the existence of a latent block structure in the message. This amounts to solving a \textit{classical} compressed sensing program
\begin{equation}\label{eq:eve_OptimalEstimator}
	\hat{\bx}_E = \argmin_{\bx\in\bbR^n} \norm{\bx}_{0} \textrm{ such that } \left\Vert \by - \bA \bx \right\Vert_2 \leq \epsilon.
\end{equation}
The identifiability condition $\bx = \hat{\bx}_E$ of \eqref{eq:eve_OptimalEstimator} is well-understood to be related to the Restricted Isometry Property (RIP) of the measurement operator \cite{candes2008restricted}. In the case of a Gaussian matrix $\bA$, the following proposition links the asymptotic failure of~\eqref{eq:eve_OptimalEstimator} to a function of the model's parameters, translating results in \cite{blanchard2011compressed} to our context.

\begin{proposition}[Failure of Eve's decoding \cite{blanchard2011compressed}]\label{prop:failureEve}
	Suppose that $\bA$ is a matrix with i.i.d. random Gaussian entries and assume a noise-free environment $\bw = \bm{0}$. Then if
	\begin{align}\label{eq:failureEveCondition}
		\beta ={}& o\left(\log \left(\frac{1}{p} \right) \right)
	\end{align}
	holds in the limit where $n \to \infty$, then the solution $\hat{\bx}_E$ of~\eqref{eq:eve_OptimalEstimator} is different from $\bx$ with overwhelming probability.
\end{proposition}

Altogether, Propositions \ref{prop:successBob} and \ref{prop:failureEve} suggest that, given the dimensions $m$ and $n$ of $\bA$, Alice can select the parameters $\beta$ and $d$ so that \eqref{eq:BobSuccessRegime} and \eqref{eq:failureEveCondition} are jointly satisfied, which is summarised in the sequel,
\begin{figure}[t]
	\centering
	\includegraphics[width=0.9\columnwidth]{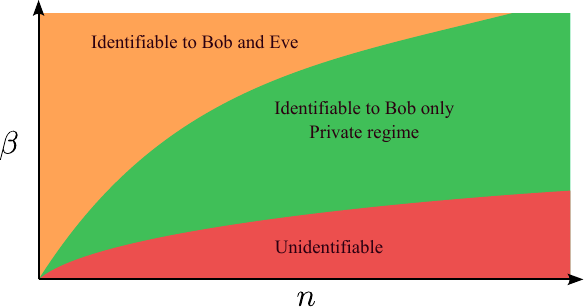}
	\caption{Regions of (non)identifiability for Eve and Bob in the single snapshot case for a block-length $d = n \log^{-\delta}(n)$ with $\delta > 0$.}
	\label{fig:sing_id}
\end{figure}

\begin{corollary}[Single snapshot privacy]\label{cor:singleSnapshotPrivacy} If Alice selects a diverging redundancy parameter $\beta \to \infty$ with $\beta = o\left(\log \left(\frac{1}{p} \right) \right)$ and $\log \left(\frac{1}{p} \right) = o \left(\frac{d}{\log(d)} \right)$, then the protocol is asymptotically private to the exchange of a single message in the limit $n \to +\infty$.
\end{corollary}
As an example, we discuss the scaling law of the parameters when the number of observation $m$ is fixed while the channel input $n$ gets large, and assume Alice allows the block length to grow with the channel input $n$ at a rate $d \sim n \log^{-\delta}(n)$ for some $\delta > 0$.
In this setup, Proposition~\ref{prop:successBob} ensures the region $\log^{\gamma}(n) \ll \beta \ll n \log^{-\delta-1}(n)$ is identifiable to Bob for any $\gamma > 0$, while Proposition~\ref{prop:failureEve} indicates $\beta \ll \log(n)$ is non-identifiable to Eve. Hence, $\log^{\gamma}(n) \ll \beta \ll \log(n)$ is asymptotically private.
This result suggests that parameter intervals for the private regime are increasing with the channel length. This highlights that the proposed communication protocol benefits from larger channel dimensions. Larger channel dimensions can be realized in practice by selecting longer spreading sequences in CDMA systems or increasing the number of antennae in MIMO systems.
In practice, Alice wants to maximize the quantity of information transmitted to Bob in a single message by transmitting messages with a maximum number of non-zero entries while remaining in the private regime. In the above example, this is achieved by selecting $p \sim n^{-1} \log^{-\gamma}(n)$.

Figure~\ref{fig:singleSnapshotNoNoise} shows the success rate of Bob and Eve to recover $\bx$ via the Block-BP and BP algorithms, respectively, for different values of the ratio $\beta$.  We see that as $\beta$ gets small, the success rates for both Bob and Eve diminish. This is intuitive as $\beta = \frac{m}{pn}$ measures the number of observations relative to the number of active components.  The lower the activity level, the fewer non-zero signals that are sent.  However, it is also clear that given $m=200$, there is a sweet spot at $\beta = 3.3$, where Bob achieves good performance while Eve does not.

\begin{figure}[t]
	\centering	\includegraphics[width=0.97\columnwidth]{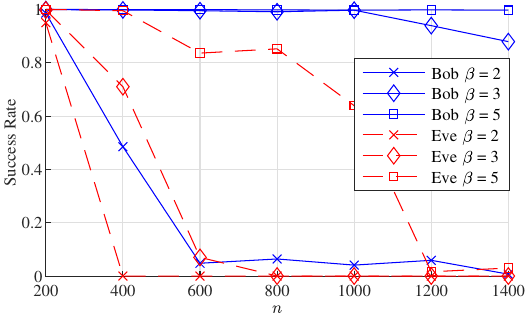}
	\caption{Success Rate of Bob and Eve to recover $\bx$ for different values of $\beta$ in the absence of noise. The parameters are set to $m = 200$, while $r$ is set to the divisor of $n$ closest to $\log_{10}^2(n)$ (\emph{i.e.} $r \simeq \log_{10}^2(n)$ and $d \simeq n \log_{10}^{-2}(n)$). The results are averaged over 1000 trials.}
	\label{fig:singleSnapshotNoNoise}
\end{figure}

\section{Eavesdropping via Higher Order Moments}\label{sec:eavesdroppingHighOrderMoments}

\subsection{Structure of the Moments}\label{sec:perfectCovariance}
In effect, our results above are for a single-pad key, \textit{i.e.} a new block structure is created for each message to be sent~\cite{bianchi2015analysis}.
To reduce the usage of the secure channel, we want to understand the reusability of $\cB$ in transmitting several independent signals $\{\bx_1, \dots, \bx_L \}$. In this scenario, if Eve can acquire multiple snapshots of observation $\{\by_1, \dots, \by_L \}$ given by $\by_\ell = \bA \bx_\ell+\bw_\ell$, $\ell = 1,\dots, L$ and under the knowledge of the prior distribution \eqref{eq:blockBernouilliGaussian} of $\bx$, she can attempt to gain statistical information about $\cB$ without having to reconstruct the messages by studying the posterior distribution of $\by$. When additional structure can be assumed for the block structure, such as contiguity of the blocks, intra-block correlation can be exploited in the block-sparse Bayesian learning framework (BSBL) to attempt to recover both the block structure and the transmitted messages simultaneously~\cite{zhang2013ExtensionSBL,fang2013PatternCoupledSparse}. In the absence of additional information on the block structure, Eve can study the moments of the posterior distribution of $
\by$. In particular, we observe that given our block signaling, the mean $\rE[\bx] = \bm{0}_N$ and covariance $\bSigma_{\bx} = p \bI_n$ of $\bx$ carry no information about the block structure $\cB$.  However, the even fourth-order moments of $\bx$ \textbf{do} provide information about the block structure, $\cB$, as seen below:
\begin{align}\label{eq:fourthOrderMomentOfX}
	\bSigma_{\bx \odot \bx}(i, j) &= \rE[x_i^2 x_j^2] -\rE[x_i^2]\rE[x_j^2] \nonumber \\
    &= \begin{cases}
		3p - p^2  & \text{if } i = j \\
		p - p^2  & \text{if } \cB(i) = \cB(j) \text{ and } i \neq j \\
		0  & \text{if } \cB(i) \neq \cB(j),
	\end{cases}
\end{align}
Additionally, as the odd fourth-order moments of $\bx$ equal zero, the terms in~\eqref{eq:fourthOrderMomentOfX} are the moments of smallest order containing information about the block structure $\cB$. As the number of samples that is necessary to estimate moments increases with their order, Eve can restrict herself to the study of the covariance $\bSigma_\bz$ of the vector $\bz = \left( \bA^{\sT} \by \right) \odot \left( \bA^{\sT} \by \right)$ in attempt to eavesdrop $\cB$ from the observation of the channel output. Given this observation, understanding the reusability of the block structure is equivalent to understanding Eve's capability to learn these fourth moments.

For notational convenience, let $\bM = \bA^\sT \bA$ and $\bP = \bM \odot \bM$. Moreover, we define the matrices $\bE_\cB$, $\bF$, and $\bG$ where each component is given, respectively, by,
\begin{subequations}\label{eq:matrixVariables-def}
	\begin{align}
		\bE_\cB(i,j) ={}& \sum_{k} \sum_{\substack{{k^\prime} \neq k \\ \cB({k^\prime}) = \cB(k)}} m_{i,k} m_{i,{k^\prime}} m_{j,k} m_{j,{k^\prime}}, \\
		\bF(i,j) ={}& \sum_{k} \sum_{{k^\prime} \neq k} m_{i,k} m_{i,{k^\prime}} m_{j,k} m_{j,{k^\prime}}, \\
		\bG(i,j) ={}& \sum_{k} \sum_{{k^\prime} \neq k} a_{k,i} a_{k,j} a_{{k^\prime},i} a_{{k^\prime},j}.
\end{align}
\end{subequations}
The next proposition, whose proof is presented in Appendix~\ref{sec:proof_covariance}, gives an expression of the covariance $\bSigma_\bv$ as a function of the matrices $\bE_\cB$, $\bF$, and of the block structure matrix $\bB$.

\begin{proposition} \label{prop:covariance} Let $\bz = \left( \bA^{\sT} \by \right) \odot \left( \bA^{\sT} \by \right)$. If $\bx$ is drawn according to \eqref{eq:blockBernouilliGaussian} then the covariance $\bSigma_\bz$ of $\bz$ is given by
	\begin{equation} \label{eq:covariance}
		\bSigma_\bz = p(1-p) \bP \bB \bP + 2p \bE_{\cB} + \bC
	\end{equation}
	where the matrix $\bC$ is given by
	\begin{multline}\label{eq:C-expression}
		\bC = 2p  \bP^2 +  2 p^2 \bF
		+ 2 p \sigma^2 \bm{M}^2 \odot \bM \\
	 +2\sigma^4\left( \left(\bA \odot \bA\right)^{\sT} \left( \bA \odot \bA \right) + \bG \right)
	\end{multline}
\end{proposition}
Proposition~\ref{prop:covariance} proposes an decomposition of the covariance matrix $\bSigma_{\bz}$ into two main terms:
\begin{enumerate}
	\item The term $p(1-p) \bP \bB \bP + 2p \bE_\cB$, which captures properties of the block structure $\cB$.
	\item The term $\bC$ which only depends on the block activation probability $p$, on the channel $\bA$, and on the noise power~$\sigma^2$.
\end{enumerate}
In the sequel, we propose a strategy by which to exploit this structure to learn $\cB$.

\subsection{Reconstruction via Spectral Clustering}

In this section, we propose a provable spectral clustering-based algorithm for Eve to infer the block structure $\cB$ from observing a finite number of snapshots $L$. In our setting, the block structure matrix $\bB$ that Eve aims to recover has a rank equal to the number of blocks $r$, which is assumed to be much smaller than the ambient signal dimension $n$. As a result, the reliability of spectral clustering can be anticipated for inferring the low-dimensional block structure.

We first review Algorithm~\ref{alg:eavesdroppingMomentMethod}. This is a straightforward algorithm that employs the matrix $\bY$, whose columns $\{\by_1, \dots, \by_L\}$ are sampled from the channel output, to determine an estimate $\hat{\bSigma}_\bz$ of the covariance matrix $\bSigma_\bz$ in Equation~\eqref{eq:covariance}. This equation is consecutively ``inverted'', yielding an estimator $\tilde{\bB}$ of the indicator matrix $\bB$. As the $r$-leading eigenvectors of the indicator matrix $\bB$ identify exactly the block structure $\cB$, an estimate $\hat{\cB}$ of the true block structure $\cB$ is constructed by clustering the rows of the $r$ leading eigenvectors of the matrix $\tilde{\bB}$, following a $K$-means-type procedure described by Algorithm~\ref{alg:SpectralClustering}.

The rest of this section is dedicated to the theoretical analysis of the estimation procedure proposed by Algorithm~\ref{alg:eavesdroppingMomentMethod}. Under incoherence assumptions on the channel matrix $\bA$, we first assess Eve's capability to eavesdrop $\cB$ using Algorithm~\ref{alg:SpectralClustering} when she has access to infinitely many channel outputs ${\by_\ell}$, and thus to the ground truth covariance matrix $\bSigma_\bz$. Then, we consider the case where Eve observes a finite number of channel outputs.

\begin{algorithm}[t]  \caption{Eavesdropping by Spectral Clustering
}\label{alg:eavesdroppingMomentMethod}
  \begin{algorithmic}[1]
    \Function{MomentMethod}{\mbox{$\bY \hspace{-0.05in}\in \bbR^{m\times L},\bA \in \bbR^{m\times n},p,r$}}
      \Let{$\bZ$}{$\left(\bA^\sT  \bY\right) \odot \left(\bA^\sT  \bY\right)$}
      \Let{$\overline{\bz}$}{$  p \diag \left({(\bA^\sT \bA)}^2 \right) + \sigma^2 \diag \left(\bA^{\sT} \bA \right) $}
      \Let{$\gamma$} {$\frac{2(d-1)}{m} + \frac{(n-2)(d-1)}{m^2}$ with $d = \frac{n}{r}$}
			\Let{$\hat{\bSigma}_\bz$} {$\frac{1}{L} \sum_{\ell=1}^L \left(\bz_\ell - \overline{\bz} \right) \left(\bz_\ell - \overline{\bz} \right)^{\sT}$}
			\Let {$\bK$} {$\bC + 2p\gamma \bI_n$} \Comment{With $\bC$ as in \eqref{eq:C-expression}}
			\Let{$\tilde{\bB}$} {$\left(p(1-p)\right)^{-1} {\bP^{-1}} (\hat{\bSigma}_\bz  - \bK) \bP^{-1}$}
			\Let {$\tilde{\bU}$} {the $r$ dominant eigenvectors of $\tilde{\bB}$}
			\Let{$\hat{\cB}$} {{\sc GreedyKMeans} $(\tilde{\bU},r)$}
      \State \Return{$\hat{\cB}$}
    \EndFunction
  \end{algorithmic}
\end{algorithm}

\begin{algorithm}
  \caption{Greedy implementation of K-means}
    \label{alg:SpectralClustering}
  \begin{algorithmic}[1]
    \Function{GreedyKMeans}{$\tilde{\bU}\in \bbR^{n\times r}$}
            \Let{$\tilde{r}$}{$0$}
			\For{$j = 1 \dots n$}
				\If{$\min_{q \in \{1,\dots,\tilde{r}\}} \norm{\tilde{\bc}_q - \tilde{\bu}_j}_2 < \frac{1}{\sqrt{2d}}$}
				\Let{$\hat{\cB}(j)$}{$\argmin_q \norm{\tilde{\bc}_q - \tilde{\bu}_j}_2$}\\
                \hfill \Comment{ Assign $j$th entry to cluster with closest centroid}
				\Let{$\tilde{\bc}_q$}
                {$\mbox{mean}\left\{\tilde{\bu}_q; \; q \leq j \textrm{ and } \hat{\cB}(q) = \hat{\cB}(j) \right\}$} \\
                \hfill \Comment{Update the centroid}
				\Else
				\Let{$\tilde{r}$}{$\tilde{r}+1$} \Comment{Create a new cluster}
				\Let{$\tilde{\bc}_{\tilde{r}}$}{$\tilde{\bu}_j$} \Comment{Assign $j$th entry to new cluster}
				\EndIf
			\EndFor
            \State \Return{$\hat{\cB}$}
    \EndFunction
  \end{algorithmic}
\end{algorithm}

\subsection{Conditions for Exact Clustering}
Eve's ability to estimate $\tilde{\bB}$ sufficiently close to the indicator matrix $\bB$ is a determining factor in her attempt to recover $\cB$. When the spectral distance $\left\Vert \tilde{\bB} - \bB \right\Vert_2$ is small enough, the eigenvectors of $\tilde{\bB}$ will align with those of $\bB$, and the block structure will become identifiable by spectral clustering.
We start the theoretical derivations by finding in Proposition~\ref{prop:clustering_errors}, a sufficient condition on $\left\Vert \tilde{\bB} - \bB \right\Vert_2$ under which the K-means clustering procedure described by Algorithm~\ref{alg:SpectralClustering} returns the exactly the secret block structure $\cB$.

\begin{proposition}[Exact clustering] \label{prop:clustering_errors} Assume $\bB \in {\{0,1\}}^{n \times n}$ is the indicator matrix of a block structure $\cB$ with  $d \geq 2$. Then, for any $\tilde{\bB} \in \bbR^{n\times n}$ with $\norm{  \tilde{\bB} - \bB }_2 < \frac{\sqrt{2d}}{8}$, the output of Algorithm \ref{alg:eavesdroppingMomentMethod} applied the matrix $\tilde{\bU}\in \bbR^{n\times r}$ that is composed by the $r$ leading eigenvector of $\tilde{\bB}$ exactly recovers the block structure, {i.e.} $\hat{\cB} = \cB$.
\end{proposition}

\begin{proof}
	First, it is easy to confirm from Equation~\eqref{eq:definition_blockmatrix} that $\norm{\bB}_2 =d$. As both $\bB$ and $\tilde{\bB}$ are Hermitian matrices, they have orthogonal bases of eigenvectors. We write $\bU,\tilde{\bU} \in \bbR^{n \times r}$ the matrices whose columns are the eigenvectors corresponding to the $r$ leading eigenvalues of $\bB$ and $\tilde{\bB}$, respectively. By the Davis-Kahan eigenvector perturbation theorem~\cite{davis1970RotationEigenvectors}, there exists an orthogonal matrix $\bO \in \bbR^{r \times r}$ such that
\begin{align} \label{eq:boundDeviation-2}
	\norm{\tilde{\bU} - \bU\bO}_{\sF} \leq& \frac{\sqrt{2}\norm{\tilde{\bB} - \bB}_2}{\lambda_{r}(\bB) - \norm{\tilde{\bB} - \bB}_2 } \nonumber \\
	={}& \frac{\sqrt{2}\norm{\tilde{\bB} - \bB}_2}{d - \norm{\tilde{\bB} - \bB}_2 } \nonumber \\
	\leq{}& \frac{2}{d} \norm{\tilde{\bB} - \bB}_2 < \frac{\sqrt{2}}{4\sqrt{d}},
\end{align}
where we used $\frac{t}{1-t} \leq \sqrt{2}t$ when $0\leq t \leq \frac{1}{4}$ in the second inequality.
Next, we denote by $\bu_j$ and $\tilde{\bu}_j$ the $j$th columns of the matrices $\bU^{\sT}$ and $\tilde{\bU}^{\sT}$, respectively. From the expression~\eqref{eq:definition_blockmatrix} of $\bB$, the vector $\bu_j$ indicates the block in which the $j$th element belongs, more precisely we have
\begin{equation}
	\bu_j(q) = \begin{cases}
		\frac{1}{\sqrt{d}} & \text{if } \cB(j) = q \\
		0 & \text{otherwise}.
	\end{cases}
\end{equation}
Suppose that $\cB(j)=q$ and let  $\bc_q = \bc_{\cB(j)} = \bO^{\sT} \bu_j$, which represent the rotated true centroid of the $q$th block. Equation~\eqref{eq:boundDeviation-2} implies that $\norm{\tilde{\bu}_j - \bc_q}_2 < \frac{\sqrt{2}}{4 \sqrt{d}}$. Therefore, this also implies that the estimated centroid of the $q$th block $\tilde{\bc}_q$ satisfies $\norm{\tilde{\bc}_q - \bc_q}_2 < \frac{\sqrt{2}}{4\sqrt{d}}$ at each step of the algorithm.
From the triangle inequality, we have,
\begin{equation}\label{eq:distance_true_cluster}
    \norm{\tilde{\bu}_j - \bc_q}_2 \leq \norm{\tilde{\bu}_j - \tilde{\bc}_q}_2 + \norm{\tilde{\bc}_q - \bc_q}_2 < \frac{1}{\sqrt{2d}}.
\end{equation}
By orthogonality of the eigenvectors $\bu_{q}$ and $\bu_{q^\prime}$, we also have that $\norm{\bc_{q^\prime} - \bc_{q}}_2 = \norm{\bm{O}^{\sT} \left( \bu_{q^\prime} - \bu_{q}\right)}_2 = \sqrt{\frac{2}{d}}$ for any $q \neq q^\prime$. Hence if $q^\prime \neq \cB(j)$ we may write
\begin{align}\label{eq:distance_wrong_cluster}
	\norm{\tilde{\bu}_j - \tilde{\bc}_{q^\prime}}_2 & = \norm{\tilde{\bu}_j - \bc_q + \bc_q - \tilde{\bc}_{q^\prime}}_2 \nonumber  \\
    &\geq \norm{\bc_q - \tilde{\bc}_{q^\prime}}_2 - \norm{\tilde{\bu}_j - \bc_q}_2  \nonumber \\
    &> \sqrt{\frac{2}{d}} - \frac{\sqrt{2}}{2\sqrt{d}} = \frac{1}{\sqrt{2d}}.
\end{align}
Hence $\norm{\tilde{\bu}_j - \tilde{\bc}_{q}}_2 < \norm{\tilde{\bu}_j - \tilde{\bc}_{q^\prime}}_2$ for any $q \neq q^\prime$, and we conclude with~\eqref{eq:distance_true_cluster} and~\eqref{eq:distance_wrong_cluster} that at the $j$th iteration, Algorithm~\ref{alg:SpectralClustering} associate $\hat{\cB}(j) = \cB(j) =q$ if there was an element in $\{1,\dots,j-1\}$ that is in the $q$th cluster, otherwise associate $j$ to a new cluster $q$. This results in $\hat{\cB} = \cB$ at the algorithm's output.\end{proof}

\subsection{Asymptotic Vulnerability}\label{ssec:asymptotic_vulnerability}

In this subsection, we assume that Eve can sample infinitely many channel output $\{\by_\ell\}$ that have been produced with \emph{the same} secret block structure $\cB$, and wish to understand Eve's capability to recover $\cB$ from Algorithm~\ref{alg:eavesdroppingMomentMethod}. Of particular interest, Eve knows in this setting the probability distribution $\by$ and consequently has access to the ground truth covariance matrix $\bSigma_\bz$ given in~\eqref{eq:covariance}.
In the additional pessimistic hypothesis where Eve knows the activation probability\footnote{In more practical considerations, the transmission parameter $p$ can be estimated by Eve from the covariance $\bSigma_\by$ of the channel output as $\bSigma_\by = p \bM$.}
$p$, the block length $d$, the channel matrix $\bA$, and the statistics of the noise $\bw$, she can compute the matrices $\bP$ and $\bC$ in Proposition~\ref{prop:covariance}, and the constant $\gamma$ defined in the fourth step of Algorithm~\ref{alg:eavesdroppingMomentMethod}. Hence, she can formulate the estimate $\tilde{\bB}$ of the block structure $\bB$ as
\begin{align} \label{eq:trueCov-1}
	\tilde{\bB} ={}& \frac{1}{p(1-p)} \bP^{-1} \left( \bSigma_\bz - 2p \gamma \bI_n - \bC \right) \bP^{-1} \nonumber \\
	={}& \bB + 2 {(1-p)}^{-1} \bP^{-1} \left(\bE_{\cB} - \gamma \bI_n \right) \bP^{-1},
\end{align}
and achieves a spectral distance to the ground truth indicator matrix
\begin{align}\label{eq:eve_spectral_distance_infinity}
	\norm{\tilde{\bB} - \bB}_2 &=  2 {(1-p)}^{-1} \norm{\bP^{-1} \left( \bE_\cB - \gamma \bI_n \right) \bP^{-1}}_2.
\end{align}
The crux is to understand whenever~\eqref{eq:eve_spectral_distance_infinity} matches the sufficiency criteria of Proposition~\ref{prop:clustering_errors} to access Eve's perfect recovery $\cB$, and the vulnerability of the proposed scheme.

To that end, we must note that the matrices $\bE_\cB, \bF$, and $\bG$ introduced in \eqref{eq:matrixVariables-def} are summations of fourth-order moments of the matrices $\bM$ and $\bA$. Furthermore, even if the entries of the matrix $\bA$ are assumed to be drawn \emph{i.i.d.}, the products considered in \eqref{eq:matrixVariables-def} are coupled, and the summations are over dependent terms. As a result, additional statistical assumptions on the distribution of the channel matrix $\bA$ are needed to control the estimate of the block structure $\tilde{\bB}$. Therefore, we provide Definition \ref{def:coherence}, which introduces a new notion of coherence relevant to our spectral clustering context.
\begin{definition}[Coherence]\label{def:coherence}
	For an $m\times n$ matrix $\bA$, we let $\bM = \bA^\sT \bA$ and $\bP = \bM \odot \bM$. Given two positive numbers $\mu > 0$ and $\nu > 0$, a matrix $\bA$ is said to be $(\mu,\nu)$-\emph{coherent} if and only if the following bounds holds:
    \begin{enumerate}
    \item \underline{First-order bounds:}
    \begin{subequations}\label{eq:coherence-def}
	\begin{align}\label{eq:coherence-A}
		\left\Vert \bA \right\Vert_2 & \leq \sqrt{\frac{n}{m}}\mu \\
        \left\Vert \bA \right\Vert_{\max} & \leq \sqrt{\frac{n \log(n)}{m}}\mu \label{eq:coherence-A_max}
	\end{align}
    \item \underline{Second-order bound:}
    \begin{align}\label{eq:coherence-M}
        \norm{\bM}_{\max} & \leq \log(n) \mu^2
    \end{align}
    \item \underline{Fourth-order bounds:}
        For any block structure $\cB$ over $n$ element with maximal block length $d$, and for $(i,j) \in \{1,\dots,n\}^2$, the fourth order matrix $\bE_\cB$ satisfies
	\begin{equation}\label{eq:coherence-E}
		\norm{\bE_\cB - \gamma \bI_n}_2 \leq  \max\left\{\frac{1}{m^2}, \frac{n}{m^4}\right\} d\sqrt{n} \log(n)  \mu^8
	\end{equation}
	where $\gamma = \frac{2(d-1)}{m^2} + \frac{(n-2)(d-1)}{m^4}$, and the fourth order matrices $\bF$ and $\bG$ satisfy
	\begin{align}\label{eq:coherence-F}
		\norm{\bF}_2 \leq{}& \frac{n^2}{m^2} \log^{2}(n) \mu^8\\
		\norm{\bG}_2 \leq{}& \frac{n}{m} \log(n) \mu^4;\label{eq:coherence-G}
	\end{align}
    The matrix $\bP$ is invertible and
    \begin{equation}\label{eq:coherence-P}
        \lambda_{\min}(\bP) \geq \nu^{-1}.
    \end{equation}
\end{subequations}
\end{enumerate}
\end{definition}

The parameter $\mu$ is raised to different exponents in~\eqref{eq:coherence-def} to maintain homogeneity across the different matrix norms. Understanding when a matrix $\bA$ is $(\mu,\nu)$-coherent is crucial to apply our theoretical analysis of Algorithm~\ref{alg:eavesdroppingMomentMethod}. However, finding coherence parameters when assuming the entries $\{a_{i,j}\}$ of $\bA$ to be drawn \emph{i.i.d.} from a know prior distribution can be particularly challenging as the quantities defined in~\eqref{eq:matrixVariables-def} are summations of fourth and eighth-order terms in the matrix $\bA$. As a result, the terms in those summations are \emph{dependent}, and the usual incoherence bounds for matrix sensing~\cite{candes2006robust,davenport2016overview} cannot be directly applied.

Nonetheless, an interesting class of matrix $\bA$ to consider is the one whose columns are drawn \emph{i.i.d.} according to a \emph{unitary} and \emph{isotropic} distribution. In that case, we have
	\begin{equation}\label{eq:EP_expr}
		\rE[\bP] = \bI_n + \frac{1}{m} \left(\bJ_n - \bI_n \right).
	\end{equation}
Under the additional assumption that the columns of $\bA$ have a \emph{bounded inner product}, \emph{i.e.} if there exists a small enough $\varepsilon > 0$ such that
\begin{align}\label{eq:bounded_inner_product}
 p_{i,j} =\left\vert \left\langle \ba_i, \ba_j \right\rangle \right\vert^2 \leq \begin{cases}
    (1+\varepsilon) & \text{if } i=j \\
    \frac{1}{m}(1+\varepsilon) & \text{if } i\neq j
\end{cases}
\end{align}
for all $(i,j)$, then we can show that $(\mu,\nu)$-coherence holds with high probability. Indeed, \eqref{eq:coherence-A} holds because of the unitary isotropic assumption on $\bA$,~\eqref{eq:coherence-A_max} is induced by the bounded, hence sub-Gaussian concentration of the matrix $\bA$ (see \emph{e.g.}~\cite{tropp2015introduction}),~\eqref{eq:coherence-M} is immediate from~\eqref{eq:bounded_inner_product}, and \eqref{eq:coherence-P} occurs from $\norm{\bP - \rE[\bP]}_2 < \lambda_{\min}(\rE[\bP]) = 1 - m^{-1}$ given a small enough $\varepsilon$. Finally, Lemma~\ref{lem:concentrationEB} validates~\eqref{eq:coherence-E} with high probability, and its proof is provided in Appendix~\ref{sec:concentration_EB}. The proofs of the two later bounds~\eqref{eq:coherence-F} and~\eqref{eq:coherence-G} are omitted for brevity and can be re-derived by following analogous reasoning.
\begin{lemma}[Concentration of $\bE_\cB$]\label{lem:concentrationEB}
	Suppose that the columns of $\bA$ are drawn \emph{i.i.d.} according to a unitary isotropic random distribution and that~\eqref{eq:bounded_inner_product} holds for some $\varepsilon > 0$. There exists a constant $\mu>0$ such that~\eqref{eq:coherence-E} is satisfied with probability greater than $1-2n^{-1}$.
\end{lemma}

A numerical validation of the $(\mu,\nu)$-coherence assumption is presented in Section~\ref{sec:numericalSimulation} when the channel matrix $\bA$ is \emph{i.i.d.} Gaussian, or with columns drawn $\emph{i.i.d.}$ uniformly on the sphere.

The $(\mu,\nu)$-coherence assumption on the matrix $\bA$ can be exploited with $p \leq \frac{1}{2}$ to control the spectral distance~\eqref{eq:eve_spectral_distance_infinity} as
\begin{align}\label{eq:spectral_distance_2}
	\norm{\tilde{\bB} - \bB}_2 &=  2 {(1-p)}^{-1}\norm{\bP^{-1} \left( \bE_\cB - \gamma \bI_n \right) \bP^{-1}}_2 \nonumber \\
	& \leq 4 \norm{\bP^{-1}}_2^2 \norm{ \bE_\cB - \gamma \bI_n }_2 \nonumber\\
	& \leq 4 \max\left\{\frac{1}{m^2}, \frac{n}{m^4}\right\} d\sqrt{n} \log(n) \nu^2 \mu^8.
\end{align}
A direct application of Proposition~\ref{prop:clustering_errors} with~\eqref{eq:spectral_distance_2} yields the following characterization of the asymptotic vulnerability of the communication protocol proposed in Section~\ref{sec:secretBlockSparsity} from an eavesdropper attempting to learn the secret block structure $\cB$ via Algorithm~\ref{alg:SpectralClustering}.
\begin{corollary}[Asymptotic vulnerability]
    Suppose that $\bA$ is $\left(\mu,\nu\right)$-coherent, then if
    \begin{equation}\label{eq:asymptotic_vulnerability}
     \delta \coloneqq \frac{\sqrt{2}}{16} \nu^{-2} \mu^{-8} - 2 \max \left\{\frac{1}{m^2}, \frac{n}{m^4} \right\} \sqrt{nd} \geq 0
    \end{equation}
    Eve can recover the block structure $\cB$ by applying Algorithm~\ref{alg:eavesdroppingMomentMethod} provided access to infinitely many samples of the channel outputs $\{\by_1, \by_2, \dots\}$.
\end{corollary}
This result suggests that the communication protocol between Alice and Bob proposed in Section~\ref{sec:secretBlockSparsity} is compromised from the knowledge of the ground truth covariance and to a channel of large enough output dimension $m$ with constant reuse of the secret key. We call this regime \emph{asymptotic vulnerability}.

\subsection{Estimation with a Finite Number of Snapshots}

In practice, Eve can access a limited number of snapshots $L$ before Alice terminates the communication or refreshes the structure $\cB$. Consequently, the true covariance $\bSigma_\bz$ always remains unknown to Eve. Instead, she can attempt to estimate $\cB$ from the empirical estimator of the covariance given by
\(
\hat{\bSigma}_{\bz} = \frac{1}{L} \sum_{\ell=1}^L \left(\bz_\ell - \rE[\bz] \right) \left(\bz_\ell - \rE[\bz] \right)^{\sT},
\)
where $\rE[\bz] = p \diag \left ( \bM^2 \right)$ and $\diag(\cdot)$ is the operator that stacks the diagonal elements of a $n \times n$ matrix into an $n$-dimensional vector. Proposition~\ref{prop:estimateOfB} provides recovery guarantees for Eve under the proviso she accesses a large enough number of snapshots~$L$.

\begin{proposition}[Estimation with Finite Numbers of Snapshots]\label{prop:estimateOfB}
	Let the quantity  $\delta$ be as defined in~\eqref{eq:asymptotic_vulnerability} and suppose that $\delta > 0$, then there exist a constants $C > 0$ such that if $L$ satisfies
	\begin{multline}\label{eq:min_L_bounds}
		\frac{\sqrt{L}}{\log(L)} \geq \delta^{-1}  \frac{n^2}{m^2}\log^2(n) \sqrt{d}  \\
        \cdot \left( 1 + \frac{2}{7\log(n)} \frac{\sigma^2}{\mu^2} + \frac{4\beta n}{7 m\log(n)}\frac{\sigma^4}{\mu^4} \right) ,
	\end{multline}
	then the output $\hat{\cB}$ of Algorithm~\ref{alg:eavesdroppingMomentMethod} satisfies $\hat{\cB} = \cB$ with probability greater than $1 - C L^{-1}$.
\end{proposition}

The proof of Proposition~\ref{prop:estimateOfB} is presented in Appendix~\ref{sec:proof_finite_snapshots}.  We observe from~\eqref{eq:min_L_bounds} that even in the absence of noise on the Alice--Eve channel, Eve still needs a non-trivial number of snapshots to recover the block structure provably.

\section{Numerical simulations}\label{sec:numericalSimulation}
Next, we will provide experiments to show the efficacy of our proposed scheme and validate theoretical results.  We underscore that our assumptions are \emph{very favorable} to Eve, who is assumed to know: (1) the channel matrix $\bA_E$, (2) the probability of block activation $p$, and  (3) the block length $d$. More practical conditions (errors in the estimate of $\bA_E$) are provided in Figure~\ref{fig:BER_vsNoisyA} and Eve's performance degrades even further.
\subsection{Coherence Assumption}
We start by validating the scaling laws of the coherence metric proposed in Definition~\ref{def:coherence}. Figure~\ref{fig:coherence} shows the empirical probability of a channel $\bA$ to be $(\mu,\nu)$-coherent for varying values of the parameters $(\mu,\nu)$. Two isotropic probability distributions are considered when the channel: 1) has \emph{i.i.d.} Gaussian entries; 2) \emph{i.i.d} columns drawn uniformly on the sphere. Under the ratio $\frac{m}{n} = \frac{1}{2}$, the numerical simulations suggest that selecting $\mu \geq 1.75$ and $\nu \leq 0.4$ (respectively, $\mu \geq 1.7$ and $\nu \leq 0.8$) is enough to guarantee $(\mu,\nu)$-coherence of the Gaussian channel  (respectively, uniform spherical channel) with high probability, independently of the channel dimensions, under proviso of a large enough $n$. Furthermore, given fixed channel dimensions, the uniform spherical channel has sharper tails than the Gaussian one, resulting in the the more favorable coherence parameters as seen in Figure~\ref{fig:coherence}.

\subsection{Validation of the Spectral Clustering Method}
In this section, we validate the theoretical findings presented in Section \ref{ssec:privacySingleSnapshot}, Section~\ref{ssec:asymptotic_vulnerability}, and Section \ref{sec:eavesdroppingHighOrderMoments} through numerical simulations. Herein, the block compressed sensing problem \eqref{eq:bob_BlockCS} and compressed sensing problem \eqref{eq:eve_OptimalEstimator} are solved using the block-basis pursuit (Block-BP) and basis pursuit (BP) convex relaxation with {\sc Matlab} and the {\sc SPGL1} package~\cite{van2009probing}. For a unitary and isometric matrix $\bA$, the signal-to-noise ratio (SNR) at the channel output is defined as $\operatorname{SNR} \triangleq \frac{\rE \left[\norm{\bA \bx}_2^2\right]}{\rE \left[ \norm{\bw}_2^2 \right]} = \frac{p n^2}{\sigma^2 m^2}$. We subsequently select $\bA$ at random with independent Gaussian entries $a_{i,j} \sim \cN \left(0, \frac{1}{m} \right)$.

\begin{figure}[t]
	\centering
	\begin{minipage}{.49\columnwidth}
		\centering
		\includegraphics[width=\textwidth]{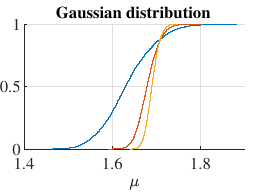} \\
 		\includegraphics[width=\textwidth]{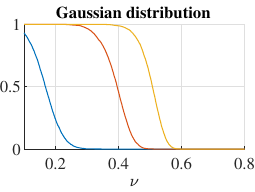}
	\end{minipage}
	\begin{minipage}{.49\columnwidth}
		\centering
		\includegraphics[width=\textwidth]{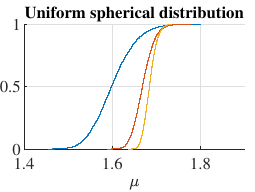} \\
   		\includegraphics[width=\textwidth]{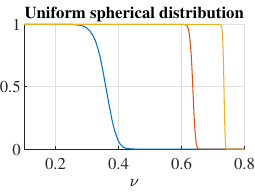}
	\end{minipage}
    \caption{The empirical probabilities of inequality~\eqref{eq:coherence-def} holding for different values of coherence parameters $(\mu,\nu)$. Top row: $\bA$ is a random Gaussian matrix with \emph{i.i.d.} entries. Bottom row: $\bA$ has columns drawn \emph{i.i.d} according to a unitary spherical distribution. In blue: $n=50$, in red: $n=200$, in yellow: $n=400$. Herein, we set $\frac{m}{n} = \frac{1}{2}$. Experiments are averaged over 5000 trials.\label{fig:coherence}}
\end{figure}

We consider the clustering capabilities of Algorithm~\ref{alg:eavesdroppingMomentMethod}. Figure~\ref{fig:cluster} shows the clusters returned by the subroutine Algorithm~\ref{alg:SpectralClustering} for different numbers of snapshots and different SNRs, for the case where $n=400, m=200$ and $\beta = 2.5$ and $r=5$; that is due to the block structure, we have $5$ clusters. It is clear that the value of $L$ (number of snapshots) impacts whether we can identify the clusters and, thus, the block structure. Additionally, high SNR values result in better identifiability of the clusters, especially under a limited number of snapshots, when the signal and the noise empirical covariances are not yet decoupled.

\begin{figure}[t]
	\centering
    \includegraphics[width=0.95\columnwidth]{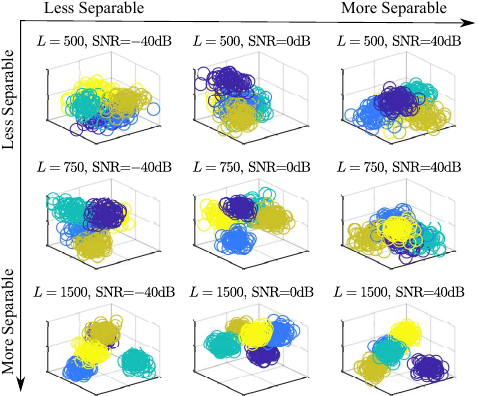}
	\caption{Projections of the clusters estimated by Algorithm~\ref{alg:SpectralClustering} unto $\bbR^3$ for different numbers of snapshots and SNRs. Rows (from top to bottom): $L=500$, $L=750$, $L=1500$. Columns (from left to right): $\operatorname{SNR}=-40\mathrm{dB}$, $\operatorname{SNR}=0\mathrm{dB}$, $\operatorname{SNR}=40\mathrm{dB}$. Other system parameters are $n=400$, $m=200$, $\beta = 2.5$ and $r=5$.}
	\label{fig:cluster}
\end{figure}

\begin{figure}[t]
	\centering
	\includegraphics[width=0.98\columnwidth]{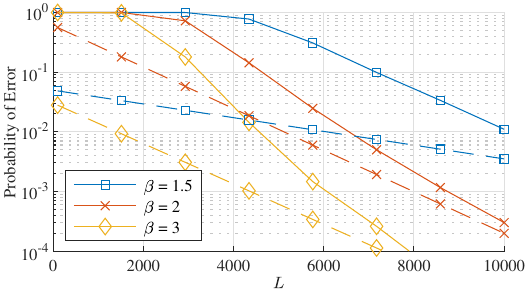}
	\caption{Probability of failure of Algorithm~\ref{alg:eavesdroppingMomentMethod} as a function of the number of snapshots $L$ for different communications rates $\beta$. Dashed lines represent Hoeffding's error rates $p_{\mathrm{Hoeff}}$ detailed in Section~\ref{sec:numericalSimulation} for the corresponding values of $\beta$. Herein, we set $n=200$, $m=100$, $r=5$, and $\operatorname{SNR} = 0\mathrm{dB}$. Experiments are averaged over $10^5$~trials.}
	\label{fig:error_exponent_noNoise}
\end{figure}

Next, we evaluate the probability for Eve to recover the correct block structure $\cB$ from the output of Algorithm~\ref{alg:eavesdroppingMomentMethod} as a function of the number of observed snapshots, $L$, that she has acquired without a refresh of the block structure. We evaluate the empirical error rate of Algorithm~\ref{alg:eavesdroppingMomentMethod}, defined by the fraction random problem instances where $\hat{\cB} \neq \cB$. To assess the secrecy of the proposed protocol, we compare this empirical error rate with the error rate of a Hoeffding test between the probability distribution $\mathcal{Y}$ of the channel output produced by the true block structure $\cB$ and the probability distribution $\mathcal{Y}^\prime$ produced by another block structure $\cB^\prime$. Given the Kullback–Leibler divergence $\KL(\mathcal{Y}, \mathcal{Y}^\prime)$ between those two distributions, Hoeffding's error rate is given by $p_{\mathrm{Hoeff}} = C \exp \left(-L\min_{\mathcal{Y}^\prime}  \left\{ {\KL(\mathcal{Y},\mathcal{Y}^\prime)}^2\right\} \right)$ for some $C>0$, where the minimum is taken over all possible block structures $\cB^\prime$ of $r$-blocks of length $d$ that is not equal to $\mathcal{Y}$. Hoeffding's error rate is an asymptotic statistical lower bound on the error probability for hypothesis testing~\cite{hoeffding1965asymptotically}. As $\mathcal{Y}$ and $\mathcal{Y}^\prime$ are Gaussian mixtures in dimension $m$ with $2^r$ classes, calculating the KL-divergence by a Monte Carlo method is computationally prohibitive, and we evaluate instead its variational approximation~\cite{hershey2007approximating}. The findings are shown in Figure~\ref{fig:error_exponent_noNoise} suggest that larger values of $\beta$ increase Eve's learning rate of the secret block structure, which corroborates with the theoretical results of Proposition~\ref{prop:estimateOfB} as $\sigma^2 \propto \frac{p n^2}{m^2}$ in fixed SNR settings.  Additionally, for larger values of $\beta$, we observe that  Algorithm~\ref{alg:eavesdroppingMomentMethod} achieves an error exponent close to Hoeffding's rate, indicating the near-optimally of the proposed moment method to eavesdrop the block structure in the asymptotic~$L\to \infty$.

\begin{figure}[t]
	\centering
	\includegraphics[width=0.98\columnwidth]{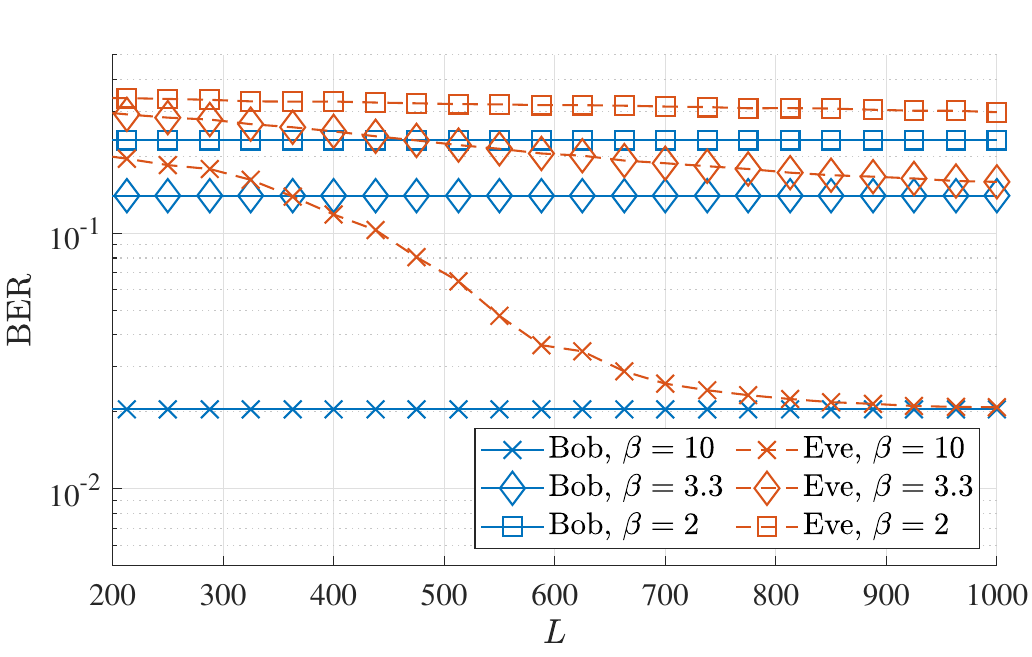}
	\caption{BER as a function of the number of snapshots $L$ for different communication rates $\beta$. Herein, we set $n=400$, $m=200$, $r=20$, and $\operatorname{SNR}=0$dB. Experiments are averaged over $10^5$ trials.}
	\label{fig:BER_vsL}
    \vspace{10pt}
 	\centering
	\includegraphics[width=0.98\columnwidth]{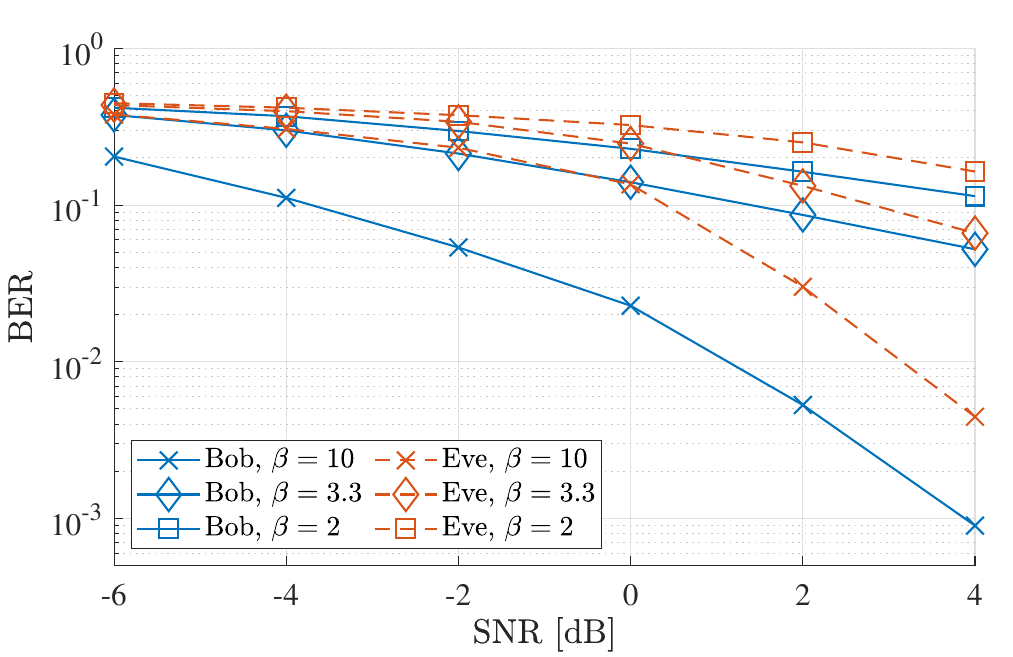}
	\caption{BER as a function of SNR for different communication rates $\beta$. Herein, we set $n=400$, $m=200$, $r=20$, and $L=400$. Experiments are averaged over $10^5$ trials.}
	\label{fig:BER_vsSNR}
\end{figure}

\subsection{Applications to MIMO systems}
Motivated by communication applications, we consider the downlink of a massive MIMO system. We assume Alice transmits parsimoniously messages encoded on a block-sparse BPSK constellation to Bob, meaning that $\bx$ is drawn according to a block-Bernoulli probability distribution, \emph{i.e.} within an active block $x_i = \pm 1$ with independent and equal probability $\frac{1}{2}$, and $x_i = 0$ within a non-active block. We define the bit-error-rate (BER) as the ratio of entries that are in the active support of Alice's message ($x_i \neq 0$) and that are incorrectly decoded by the receiver, \emph{i.e.} $\hat{x}_i \neq x_i$. Assuming that Eve relies on its estimate $\hat{\cB}$ of the block structure obtained from the output of Algorithm~\ref{alg:eavesdroppingMomentMethod}, we empirically evaluate Bob's and Eve's BERs as a function of the number of snapshots in Figure~\ref{fig:BER_vsL}, and as a function of the SNR in Figure~\ref{fig:BER_vsSNR}. The figures suggest that larger values of $\beta$ ease both Bob's and Eve's decoding. Eve can achieve the same BER as Bob if the secret structure is reused sufficiently many times. Additionally, for a fixed number of snapshots, Eve's decoding is more impeded by the noise than Bob's, and the BER margin between Bob and Eve increases with the redundancy parameter~$\beta$. Hence, for fixed channel dimensions, if Alice reduces her communication rate with Bob by selecting a smaller block activation probability $p$, she can harden Eve's decoding. This observation shows the trade-off between the communication rate Alice can achieve and the secrecy against an eavesdropper the protocol can induce. For the example considered, it is clear that for $\beta=10$ and $L<500$ Eve cannot decode while Bob can.

In practice, the channel is usually harder to estimate for Eve than for Bob. Figure~\ref{fig:BER_vsNoisyA} concludes the numerical results by comparing Bob's and Eve's BER when the eavesdropper has perfect and imperfect knowledge of $\bA_E$. Herein, the estimate $\tilde{\bA}_E$ of the channel matrix is modelled as $\tilde{\bA}_E = \bA_{E} + \bW_A$, where $\bW_A$ is an matrix with \emph{i.i.d.} Gaussian entries. The SNR on Eve's estimate of the channel is defined as $\operatorname{SNR}_{A}\triangleq \frac{\rE \left[\|\bm A\|^2_{\sF}\right]}{\rE \left[\|\bm W_{\bm A}\|^2_{\sF}\right]}$. Additionally, a baseline performance comparison is done with a classical MIMO channel of dimension $m \times pn$ where BPSK symbols $\tilde{\bx} \in \mathbb{R}^{pn}$ are transmitted without any physical layer security scheme, and retrieved by a receiver with a maximum likelihood decoder (ML). As expected, imperfect knowledge of the channel matrix further diminishes Eve's decoding performance, enhancing the privacy of the communication protocol. The baseline ML decoding illustrates the compromise between enabling privacy and achieving good error rates.

\begin{figure}[t]
 	\centering
	\includegraphics[width=0.98\columnwidth]{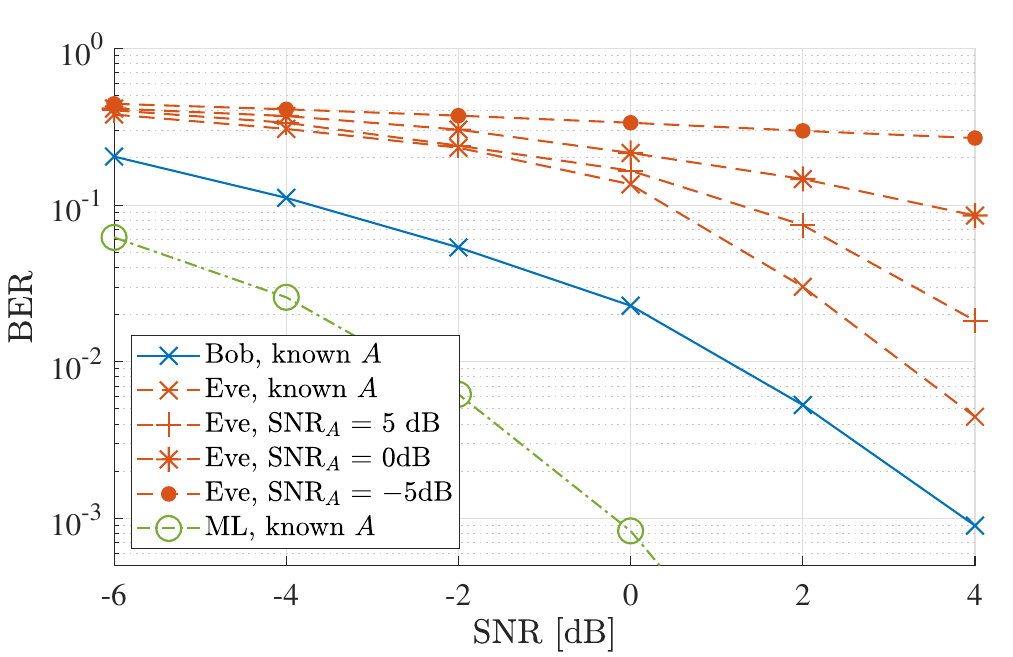}
	\caption{BER as a function of SNR. Herein, we set $\beta=10$, $n=400$, $m=200$, $r=20$, and $L=400$. For the cases with non-perfect knowledge of $\bm A$, it is assumed that Eve only has access to $\bm A + \bm W_{\bm A}$, where each element in $\bm W_{\bm A}$ is white Gaussian noise, and we define $\operatorname{SNR}_{A}\triangleq \frac{\rE \left[\|\bm A\|^2_{\sF}\right]}{\rE \left[\|\bm W_{\bm A}\|^2_{\sF}\right]}$. Experiments are averaged over $10^5$ trials.}
	\label{fig:BER_vsNoisyA}
\end{figure}

\section{Conclusions and Future Work}\label{sec:conclusion}
This article introduced a novel communication protocol with provable privacy guarantees. The proposed method harnesses a secret block-sparse prior to recovering the initial message from underdetermined linear measurements gathered at the output of a fat channel matrix. As block sparsity allows exact recovery in conditions where classical compressed sensing would provably fail, we established the existence of a secure transmission regime to a single snapshot between Alice and Bob. We studied the privacy guarantees of this communication protocol for multiple transmissions without refreshing the shared secret and proposed an algorithm for an eavesdropper to learn the block structure via the method of moments. The proposed block structure estimator appears to be asymptotically near-optimal. We validated the privacy benefits of this framework through numerical experiments.

Possible extensions of this work include a comprehensive study of the trade-off between the communication rate that Alice and Bob can achieve and the lifespan of the secret block structure. Additionally, the proposed scheme paves the way for further linear inverse problem-based implementation of private communication protocols over the physical layer.

\appendices

\section{Proof of Proposition \ref{prop:covariance}}\label{sec:proof_covariance}
Let $\bu = \bM \bx = \bA^\sT \bA \bx$ and $\tilde{\bw} = \bA^\sT \bw$ for convenience purposes. Moreover let $\bz = \left( \bA^{\sT} \by \right) \odot \left( \bA^{\sT} \by \right)$. We have that
\begin{align}\label{eq:z-1}
	\bz ={}& (\bA^\sT \by) \odot (\bA^\sT \by) = (\bu + \tilde{\bw}) \odot (\bu + \tilde{\bw}) \nonumber\\
	={}& \bu \odot \bu + 2 \bu \odot \tilde{\bw} + \tilde{\bw} \odot \tilde{\bw}.
\end{align}
We aim to derive the expression of the covariance of $\bz$. First, the independence between $\bx$ and $\bw$ implies the independence between $\bu$ and $\tilde{\bw}$. Additionally, the assumptions $\rE[\bx] = \bm{0}$ and $\rE[\bw] = \bm{0}$ imply that $\rE[\bu] = 0$ and $\rE[\tilde{\bw}] = 0$. This yields \vspace{-10pt}
\begin{subequations}
	\begin{align}
		\cov(\bu \odot \bu, \tilde{\bw} \odot \tilde{\bw}) &= \bm{0},\\
		\cov(\bu \odot \bu, \bu \odot \tilde{\bw}) &= \bm{0}, \\
		\cov(\bu \odot \tilde{\bw}, \tilde{\bw} \odot \tilde{\bw}) &= \bm{0}.
	\end{align}
\end{subequations}
Hence the covariance matrix $\bSigma_\bz = \cov(\bz,\bz)$ of the random vector $\bz$ reduces to
\begin{align}
	\label{eq:sigmaz-1}
	\bSigma_\bz = \bSigma_{\bu \odot \bu} + 2 \bSigma_{\bu \odot \tilde{\bw}} + \bSigma_{\tilde{\bw} \odot \tilde{\bw}}.
\end{align}
We derive in the sequel the expression of each of the three matrices on the right-hand side of \eqref{eq:sigmaz-1}.

\paragraph{Expression of $\bSigma_{\bu \odot \bu}$} By definition of the vector $\bu$, we have for any $i = 1, \dots, n$ that
\begin{align}
	u_i ={}& \sum_{k=1}^n x_k \inp{\ba_i}{\ba_k} = \sum_{k=1}^n x_k m_{i,k}.
\end{align}
Thus, denoting $p_{i,k}$ the $(i,k)$-th term of the matrix $\bP = \bM \odot \bM$, the expected value  $\rE[u_i^2]$ of the random variable $u_i^2$ is given by
\begin{align}\label{eq:exp_ui_square}
    \rE[u_i^2] ={}& \rE\left[\left(\sum_{k = 1}^n x_k m_{i,k} \right) \left(\sum_{{k^\prime} = 1}^n x_{k^\prime} m_{i,{k^\prime}}\right)\right] \nonumber \\
    ={}&\sum_{k=1}^n \sum_{{k^\prime} = 1}^n \rE[x_k x_{k^\prime}] m_{i,k} m_{i,{k^\prime}} \nonumber \\
    ={}&\sum_{k=1}^n \rE[x_k^2] m_{i,k}^2 = \sum_{k=1}^n \rE[x_k^2] p_{i,k}.
\end{align}
For readability, we drop the summation interval from $1$ to $n$ in the following. A direct calculation of the expected value $\rE[u_i^2 u_j^2]$ of the product $u_i^2 u_j^2$ yields
\begin{align}\label{eq:expSquares}
    \MoveEqLeft \rE[u_i^2 u_j^2] & \nonumber \\
		& \hspace{-0.3in} = \sum_{k_1} \sum_{k_1^\prime} \sum_{k_2} \sum_{k_2^\prime} \rE[x_{k_1} x_{k_1^\prime} x_{k_2} x_{k_2^\prime}] m_{i,k_1} m_{i,k_1^\prime} m_{j,k_2} m_{j,k_2^\prime} \nonumber \\
    & \hspace{-0.3in}  = \sum_{k} \sum_{{k^\prime}} \rE[x_k^2 x_{k^\prime}^2] p_{i,k} p_{j,{k^\prime}}  \nonumber \\
		& \quad + 2 \sum_{k} \sum_{{k^\prime} \neq k} \rE[x_k^2 x_{k^\prime}^2]  m_{i,k} m_{i,{k^\prime}} m_{j,k} m_{j,{k^\prime}}.
\end{align}
Equations~\eqref{eq:exp_ui_square} and~\eqref{eq:expSquares} immediately lead to an expression of the generic term of the covariance matrix $\bSigma_{\bu \odot \bu}$ of the form
\begin{align}\label{eq:cov-1}
	\MoveEqLeft \bSigma_{\bu \odot \bu}(i,j) = \rE[u_i^2 u_j^2] - \rE[u_i^2]\rE[u_j^2]& \nonumber\\
	&= \sum_{k} \sum_{{k^\prime}} \left(\rE[x_k^2 x_{k^\prime}^2] - \rE[x_k^2]\rE[x_{k^\prime}^2] \right) p_{i,k} p_{j,k^\prime} \nonumber \\
	&\qquad + 2 \sum_{k} \sum_{{k^\prime} \neq k} \rE[x_k^2 x_{k^\prime}^2] m_{i,k} m_{i,{k^\prime}} m_{j,k} m_{j,{k^\prime}} \nonumber \\
	& = \bp_i^\sT \bSigma_{\bx \odot \bx} \bp_j \nonumber \\
  & \qquad + 2 \sum_{k} \sum_{{k^\prime} \neq k} \rE[x_k^2 x_{k^\prime}^2] m_{i,k} m_{i,{k^\prime}} m_{j,k} m_{j,{k^\prime}}.
\end{align}

Next, the second sum on the right-hand side of expression \eqref{eq:cov-1} can be reformulated as
\begin{align}\label{eq:cov-2}
	\MoveEqLeft \sum_{k} \sum_{{k^\prime} \neq k} \rE[x_k^2 x_{k^\prime}^2] m_{i,k} m_{i,{k^\prime}} m_{j,k} m_{j,{k^\prime}} & \nonumber \\
	&={}\sum_{k} \sum_{\substack{{k^\prime} \neq k \\ \cB({k^\prime}) = \cB(k)}} \rE[x_k^2 x_{k^\prime}^2] m_{i,k} m_{i,{k^\prime}} m_{j,k} m_{j,{k^\prime}} \nonumber \\
	&{} \qquad + \sum_{k} \sum_{\substack{{k^\prime} \neq k \\ \cB({k^\prime}) \neq \cB(k)}} \rE[x_k^2 x_{k^\prime}^2] m_{i,k} m_{i,{k^\prime}} m_{j,k} m_{j,{k^\prime}} \nonumber\\
	& ={} p \sum_{k} \sum_{\substack{{k^\prime} \neq k \\ \cB({k^\prime}) = \cB(k)}} m_{i,k} m_{i,{k^\prime}} m_{j,k} m_{j,{k^\prime}} \nonumber \\
	&{} \qquad + p^2 \sum_{k} \sum_{\substack{{k^\prime} \neq k \\ \cB({k^\prime}) \neq \cB(k)}}  m_{i,k} m_{i,{k^\prime}} m_{j,k} m_{j,{k^\prime}} \nonumber \\
	& = p \bE_\cB(i,j) + p^2 \bF(i,j).
\end{align}
Given the fourth-order statistics~\eqref{eq:fourthOrderMomentOfX} on the block Bernouilli-Gaussian distribution of $\bx$, we have the matrix expression \(\bSigma_{\bx \odot \bx} = 2p\bI_n + p(1-p)\bB\). Substituting the previous with~\eqref{eq:cov-2} into~\eqref{eq:cov-1} yields the expression
\begin{align}\label{eq:cov-u}
	\bSigma_{\bu \odot \bu} ={}& \bP \bSigma_{\bx \odot \bx} \bP + 2 p \bE_\cB + 2 p^2 \bF \nonumber \\
	={}& p(1-p)\bP \bB \bP + 2p\bE_{\cB} + 2p\bP^2  + 2p^2 \bF.
\end{align}

\paragraph{Expression of $\bSigma_{\bu \odot \tilde{\bw}}$} From the independence assumption between $\bx$ and $\bw$, the expectation reduces to
\begin{align}\label{eq:cross-1}
	\rE[(\bu \odot \tilde{\bw})_i] ={}& \rE\left[ \left( \sum_k m_{i,k} x_k \right)\left( \sum_{k^\prime} a_{{k^\prime},i} w_{k^\prime} \right)\right] \nonumber \\
	={}&\sum_k \sum_{k^\prime} \rE[x_k] \rE[w_{k^\prime}] m_{i,k} a_{{k^\prime},i} = 0,
\end{align}
and the correlation term is
\begin{align}\label{eq:cross-2}
	\MoveEqLeft \rE[(\bu \odot \tilde{\bw})_i (\bu \odot \tilde{\bw})_j] & \nonumber \\
	& ={} \rE\left[ \left( \sum_{\ell_1} m_{i,\ell_1} x_{\ell_1} \right)\left( \sum_{k^\prime_1} a_{{k^\prime_1},i} w_{\ellp_1} \right) \right. \nonumber \\
	&{} \qquad \left. \left( \sum_{\ell_2} m_{j,\ell_2} x_{\ell_2} \right)\left( \sum_{\ell_2^\prime} a_{\ell_2^\prime,j} w_{\ell_2^\prime} \right) \right] \nonumber\\
	&={} \left(\sum_k \rE[x_k^2] m_{i,k} m_{j,k} \right) \left( \sum_{{k^\prime}} \rE[w_{{k^\prime}}^2] a_{{k^\prime},i}  a_{{k^\prime},j} \right) \nonumber \\
	&={} p \sigma^2 \left(\sum_k m_{i,k} m_{j,k} \right) \left( \sum_{{k^\prime}}  a_{{k^\prime},i}  a_{{k^\prime},j} \right) \nonumber \\
    &= p \sigma^2 \left(\sum_k m_{i,k} m_{j,k} \right)  m_{i,j}
\end{align}
and we use \eqref{eq:cross-1} and \eqref{eq:cross-2} to get
\begin{align}\label{eq:cov-uwtilde}
	\bSigma_{\bu \odot \tilde{\bw}} ={}& p \sigma^2 \bM^2 \odot \bM.
\end{align}

\paragraph{Expression of $\bSigma_{\tilde{\bw} \odot \tilde{\bw}}$} We follow analogous reasoning than for the previous term. First, the expectation is given by
\begin{align}\label{eq:corrNoise-1}
	\rE[(\tilde{\bw} \odot \tilde{\bw})_i] ={}& \rE\left[ \left( \sum_k a_{k,i} w_k \right)\left( \sum_{k^\prime} a_{{k^\prime},i} w_{k^\prime} \right)\right] \nonumber \\
	={}& \sum_k \sum_{k^\prime} \rE[w_k w_{k^\prime}] a_{k,i} a_{{k^\prime},i} \nonumber\\
	={}& \sum_k \rE[w_k^2] a_{k,i}^2 = \sigma^2 \sum_k  a_{k,i}^2.
\end{align}
Hence, the generic covariance term is
\begin{align}\label{eq:cross-3}
	\MoveEqLeft \rE[( \tilde{\bw} \odot \tilde{\bw})_i ( \tilde{\bw}\odot \tilde{\bw})_j] - \rE[(\tilde{\bw} \odot \tilde{\bw})_i] \rE[(\tilde{\bw} \odot \tilde{\bw})_j]& \nonumber \\
	&={}\rE\left[ \left( \sum_{k_1} a_{{k_1},i} w_{k_1} \right)\left( \sum_{k^\prime_1} a_{{k^\prime_1},i} w_{k^\prime_1} \right) \right. \nonumber \\
	& \qquad \left. \left( \sum_{k_2} a_{k_2,j} w_{k_2} \right)\left( \sum_{k_2^\prime} a_{k_2^\prime,j} w_{k_2^\prime} \right) \right] \nonumber \\
	&={} \sum_k \sum_{k^\prime} \left(\rE[w_k^2 w_{k^\prime}^2] - \rE[w_k^2]\rE[w_k^2] \right) a_{k,i}^2 a_{{k^\prime},j}^2 \nonumber \\
	&{} \qquad + 2 \sum_k \sum_{{k^\prime} \neq k} \rE[w_k^2 w_{k^\prime}^2] a_{k,i} a_{k,j} a_{{k^\prime},i} a_{{k^\prime},j} \nonumber \\
	&={} \sum_k \sum_{k^\prime} \left(\rE[w_k^2 w_{k^\prime}^2] - \rE[w_k^2]\rE[w_k^2] \right) a_{k,i}^2 a_{{k^\prime},j}^2 \nonumber \\
	&{} \qquad + 2 \sigma^4 \sum_k \sum_{{k^\prime} \neq k} a_{k,i} a_{k,j} a_{{k^\prime},i} a_{{k^\prime},j}.
\end{align}
Furthermore as $\bw$ is an \emph{i.i.d.} white Gaussian random vector with variance $\sigma^2$, we have that $\bSigma_{\bw \odot \bw} = 2 \sigma^4 \bI_m$. Hence, we may write
\begin{align}\label{eq:cov-wtilde}
	\bSigma_{\tilde{\bw}\odot \tilde{\bw}} ={}& \left(\bA \odot \bA\right)^{\sT} \bSigma_{\bw \odot \bw} \left(\bA \odot \bA \right) + 2 \sigma^4 \bG \nonumber\\
	={}& 	2 \sigma^4 \left( \left(\bA \odot \bA\right)^{\sT}  \left(\bA \odot \bA \right) + \bG \right).
\end{align}

We achieve the desired statement by substituting \eqref{eq:cov-u}, \eqref{eq:cov-uwtilde} and \eqref{eq:cov-wtilde} into \eqref{eq:sigmaz-1}.
\qed

\section{Proof of Proposition \ref{prop:estimateOfB}}\label{sec:proof_finite_snapshots}
We start the proof by noticing that from Proposition~\ref{prop:covariance}, the indicator matrix $\bB$ of the block structure matrix $\cB$ is given by
\begin{equation}
	\bB = \frac{1}{p(1-p)} \bP^{-1} \left( \bSigma_\bz - 2 p\bE_\cB - \bC \right) \bP^{-1},
\end{equation}
where $\bC$ is defined in \eqref{eq:C-expression} and is independent of $\cB$. The spectral distance $\norm{\tilde{\bB}- \bB}_2$ can be bounded with the triangle inequality, the $(\mu,\nu)$-incoherence of the matrix $\bA$, and the assumption $p \leq \frac{1}{2}$ as follows
\begin{align}\label{eq:B_Spectral_distance_1}
	\MoveEqLeft \norm{\tilde{\bB} - \bB}_2 &  \nonumber \\
  &\hspace{-0.3in}=\frac{1}{p(1-p)} \left\Vert \bP^{-1} \Big( \hat{\bSigma}_\bz - \bSigma_\bz + 2p(\bE_\cB - \gamma \bI_n) \Big) \bP^{-1} \right\Vert_2 \nonumber \\
  & \hspace{-0.3in} \leq \left\Vert \bP^{-1} \right\Vert_2^2 \Big( 2p^{-1} \norm{\hat{\bSigma}_\bz - \bSigma_\bz}_2  + 4 \norm{\bE_\cB - \gamma \bI_n}_2 \Big) \nonumber \\
  & \hspace{-0.3in} \leq \nu^2 \Bigg( 2 p^{-1} \norm{\hat{\bSigma}_\bz - \bSigma_\bz}_2 \nonumber  \\
	& \qquad \qquad + 4\max\left\{\frac{1}{m^2}, \frac{n}{m^4}\right\} d\sqrt{n} \log(n) \mu^8  \Bigg).
\end{align}
We obtain from~\eqref{eq:B_Spectral_distance_1} and Proposition~\ref{prop:clustering_errors} that Algorithm~\ref{alg:eavesdroppingMomentMethod} outputs the true block structure if
\begin{align}\label{eq:sigma_hat_sufficient}
    \MoveEqLeft \norm{\hat{\bSigma}_\bz - \bSigma_\bz}_2 & \nonumber \\
    & \leq p\sqrt{d} \mu^8 \left(\frac{\sqrt{2}}{16} \nu^{-2}\mu^{-8} - 2\max\left\{\frac{1}{m^2}, \frac{n}{m^4} \right\}d\sqrt{n}\right) \nonumber \\
    & \leq p\sqrt{d} \mu^{8} \delta,
\end{align}
where the right-hand side of~\eqref{eq:sigma_hat_sufficient} is non-negative by the assumption $\delta > 0$. The estimated covariance error on the left-hand side of~\eqref{eq:sigma_hat_sufficient} can be made arbitrarily small for a sufficiently large number of snapshots~$L$. Lemma~\ref{lem:concentrationCovariance} provides a high-probability bound on the error on the estimated covariance error as $L\to\infty$ in terms of the problem parameters.

\begin{lemma}[Covariance estimation]\label{lem:concentrationCovariance}
Under the hypothesis of Proposition~\ref{prop:estimateOfB}, there exist a constant $C>0$ such that the event
\begin{multline}
     \norm{ \hat{\bSigma}_\bz - \bSigma_\bz}_2 \leq  \frac{\log(L)}{\sqrt{L}} \beta^{-1} \frac{n}{m} \log^2(n) d \mu^8 \\
     \cdot \left( 1 + \frac{2}{7 \log(n)} \frac{\sigma^2}{\mu^2} + \frac{4\beta n}{7 m\log(n)}\frac{\sigma^4}{\mu^4} \right)
\end{multline}
holds with probability greater than $1 - C L^{-1}$.
\end{lemma}
For readability, the proof of Lemma~\ref{lem:concentrationCovariance} is deferred to Appendix~\ref{sec:proof_lemma_concentrationCovariance}.
It suffices to replace the left-hand side of inequation~\eqref{eq:sigma_hat_sufficient} with the high probability bound given Lemma~\ref{lem:concentrationCovariance} to yield the desired statement. \qed

\section{Proof of the Technical Lemmas}
\subsection{Proof of Lemma~\ref{lem:concentrationEB}}\label{sec:concentration_EB}

We start by studying the expected value of the diagonal terms of $\bE_\cB$. We have that
\begin{align}\label{eq:EB_diag_term}
\bE_\cB(i,i) = \sum_{k} \sum_{\substack{{k^\prime} \neq k \\ \cB(k) = \cB({k^\prime})}} p_{i,k} p_{i,{k^\prime}},
\end{align}
and we write for each $i = \{1,\dots,n\}$
\begin{equation}\label{eq:gamma_i}
\gamma_i = \sum_{k} \sum_{\substack{{k^\prime} \neq k \\ \cB(k) = \cB({k^\prime})}} \rE[p_{i,k}] \rE[p_{i,{k^\prime}}].
\end{equation}
By the isotropy assumption on the matrix $\bA$, $\gamma_i$ is constant for different values of $i$, and we may write $\gamma \triangleq \gamma_1 = \dots = \gamma_n$. Moreover, by~\eqref{eq:EP_expr}, the right hand side of~\eqref{eq:gamma_i} is a summation over $n(d-1)$ elements yielding
\begin{equation}\label{eq:ep_prod}
    \rE[p_{i,k}] \rE[p_{i,{k^\prime}}] = \begin{cases}
        \frac{1}{m} & \text{if } k = i \text{ or } k^\prime = i \\
        \frac{1}{m^2} & \text{otherwise}.
    \end{cases}
\end{equation}
Counting the number of occurrences in each case, we have
\(
    \gamma = \frac{2(d-1)}{m} + \frac{(n-2)(d-1)}{m^2}.
\)
Additionally, under the lemma's conditions,~\eqref{eq:EB_diag_term} and \eqref{eq:gamma_i} imply
\begin{align}\label{eq:Eb_diag_bound}
    \MoveEqLeft \left\vert \bE_\cB(i,i) - \gamma \right\vert = \left\vert \bE_\cB(i,i) - \gamma_i \right\vert \nonumber \\
    &=\left\vert \bE_\cB(i,i) - \sum_{k} \sum_{\substack{{k^\prime} \neq k \\ \cB(k) = \cB({k^\prime})}} \rE[p_{i,k}] \rE[p_{i,{k^\prime}}] \right\vert & \nonumber \\
    & \leq \sum_{k} \sum_{\substack{{k^\prime} \neq k \\ \cB(k) = \cB({k^\prime})}} \left\vert p_{i,k}p_{i,{k^\prime}} - \rE[p_{i,k}] \rE[p_{i,{k^\prime}}] \right\vert \nonumber \\
    &\leq \sum_k \sum_{\substack{{k^\prime} \neq k \\ \cB(k) = \cB({k^\prime})}} \left\vert p_{i,k} - \rE[p_{i,k}]\right\vert\rE[p_{i,{k^\prime}}] \left( 1 +  \varepsilon \right) \nonumber \\
    & =  \left( 1 +  \varepsilon \right) \sum_k \left\vert p_{i,k} - \rE[p_{i,k}]\right\vert \left(\sum_{\substack{{k^\prime} \neq k \\ \cB(k) = \cB({k^\prime})}} \rE[p_{i,{k^\prime}}] \right)
\end{align}
where we used in the second inequality the assumption $\left\vert p_{i,{k^\prime}} - \rE[p_{i,{k^\prime}}] \right\vert \leq \rE[p_{i,{k^\prime}}]\varepsilon$. As a result, by the isometry assumption, the terms of the summation in the right-hand side of~\eqref{eq:Eb_diag_bound} are \emph{independent} and \emph{bounded} by ${(1+\varepsilon)}$ and $\frac{{(1+\varepsilon)}}{m^2}$ when $k = i$ and $k \neq i$, respectively. Hence, the Chernoff bound can be applied~\cite{boucheron2013concentration}, and we have
\begin{multline}\label{eq:Eb_diag_chernoff}
    \bbP\left\{ \left\vert \bE_\cB(i,i) - \gamma \right\vert \leq \left(1 + \frac{d-1}{m^2}  \sqrt{2n} \log(n) \right)  \left( 1 + \varepsilon\right)^2 \right\} \\ \geq 1 - 2n^{-2}.
\end{multline}

On the off-diagonal, because of the isotropy assumption, the random variable $m_{i,k} m_{i,{k^\prime}} m_{j,k} m_{j,{k^\prime}}$ with $i \neq j$ has an even distribution for all $k$ and ${k^\prime}$  whenever $i \neq j$. Therefore its expected value is null, that is $\rE[m_{i,k} m_{i,{k^\prime}} m_{j,k} m_{j,{k^\prime}}] = 0$. Denote by $\overline{\bE}_\cB$ the matrix with off-diagonal terms equal to $\bE_\cB$ with diagonal entries $\overline{\bE}_\cB(i,i) = 0$ for all $i\in \{1, \dots, n \}$. Relying on the symmetrization principle, we introduce the Rademacher random variable $\rho_{i,j} = \sign(\bE_\cB(i,j))$, where $\sign(\cdot)$ denotes the signum function. We note that $\{\rho_{i,j}\}_{i;j\geq i+1}$  are pair-wise independent. Furthermore, we can decompose the matrix $\overline{\bE}_\cB$ as the sum
\begin{align} \label{eq:EB_decomposition}
	\overline{\bE}_\cB ={}& \sum_{i=1}^n \sum_{j=i+1}^n \rho_{i,j} \left\vert \bE_\cB(i,j) \right\vert \left(\be_i^{\sT} \be_j + \be_j^{\sT} \be_i \right).
\end{align}
Next, we recall in Proposition \ref{prop:rademacherMatrix_1} (see \emph{e.g.} \cite[Theorem 4.1.1]{tropp2015introduction}) a matrix norm concentration inequality for matrices with Rademacher entries.

\begin{proposition}[Sum of symmetric Rademacher matrix series]
 \label{prop:rademacherMatrix_1}
Consider a fixed symmetric matrix $\bB$ of dimension $n$. Let $\{\rho_{i,j}\}_{i;j\geq i+1}$ be a finite sequence of independent Rachemacher variables, and  introduce the matrix Rademacher series
\begin{equation}
	\bZ = \sum_{i=1}^n \sum_{j=i+1}^n \rho_{i,j} b_{i,j} \left(\be_i^{\sT} \be_j + \be_j^{\sT} \be_i \right).
\end{equation}
Let $v$ be the matrix variance statistic of the Rademacher sum defined as
\(
	v(\bZ) = \max_j \{ \left\Vert \bb_j \right\Vert^2_2 \}
\)
then for all $t>0$ we have
\begin{align}
	\bbP \{ \left\Vert \bZ \right\Vert_2 \geq t \} \leq{}& 2n \exp \left(\frac{-t^2}{2v(\bZ)} \right).
\end{align}
\end{proposition}
To bound $\norm{\overline{\bE}_B}_2$ using Proposition~\ref{prop:rademacherMatrix_1}, we evaluate the matrix variance $v(\overline{\bE}_\cB)$ from the decomposition~\eqref{eq:EB_decomposition}. It yields
\begingroup
\allowdisplaybreaks[2]
\begin{align}\label{eq:vEb-calculation}
\MoveEqLeft v(\overline{\bE}_\cB) = \max_j \{ \left\Vert \overline{\bE}_{\cB,j} \right\Vert^2_2 \} & \nonumber \\
        & \hspace{-0.25in} = \max_j\Bigg\{ \sum_{i \neq j} \Bigg( \sum_{\substack{{k^\prime} \neq k \\ \cB(k) = \cB({k^\prime})}}  \left\vert m_{i,k} m_{i,{k^\prime}} m_{j,k} m_{j,{k^\prime}} \right\vert \Bigg)^2 \Bigg\} \nonumber \\
         & \hspace{-0.25in} = \max_j\Bigg\{ \sum_{i \neq j} \Bigg( \sum_{\substack{{k^\prime} \neq k \\ \cB(k) = \cB({k^\prime})}}  \sqrt{p_{i,k} p_{i,{k^\prime}} p_{j,k} p_{j,{k^\prime}}} \Bigg)^2 \Bigg\} \nonumber \\
         & \hspace{-0.25in} \leq \max_j \Bigg\{ \sum_{i \neq j}  {\Bigg( \hspace{-0.08in}\sum_{\substack{{k^\prime} \neq k \\ \cB(k) = \cB({k^\prime})}}  \hspace{-0.17in} \sqrt{\rE[p_{i,k}] \rE[p_{i,{k^\prime}}] \rE[p_{j,k}] \rE[p_{j,{k^\prime}}]} \Bigg)}^2 \Bigg\} \nonumber \\
        &  \qquad  \cdot {(1+\varepsilon)}^4
\end{align}
\endgroup
where the quantity to maximize in the last inequality is constant across different values for $j$ and can be evaluated for $j=1$ without loss of generality. The inner summation in~\eqref{eq:vEb-calculation} is taken over $n(d-1)$ terms, which are equal to $\frac{1}{m^2}$ when $k \neq i$ and $k^\prime \neq i$, and equal to $\frac{1}{m}$ when $k = i$ or $k^\prime = i$. After counting the occurrences, we may reduce~\eqref{eq:vEb-calculation} to
\begin{equation}
    v(\overline{\bE}_\cB) \leq 2 n d^2 \left(\frac{1}{m} + \frac{n}{m^2}\right)^2 {(1+\varepsilon)}^4.
\end{equation}
Applying the matrix concentration inequality of Proposition~\ref{prop:rademacherMatrix_1} with $t = 2 \sqrt{ n \log(n)}  d \left(\frac{1}{m^2} + \frac{n}{m^4}\right) {(1+\varepsilon)}^2 $ induces
	\begin{multline}\label{eq:bE_offdiag_prob}
	    \bbP \left\{ \left\Vert \overline{\bE}_\cB \right\Vert_2 \leq 2 \sqrt{2 n \log(n)}  d \left(\frac{1}{m^2} + \frac{n}{m^4}\right) {(1+\varepsilon)}^2 \right\} \\ \geq{} 1- 2 n^{-1}.
	\end{multline}

We are now ready to achieve the desired statement. First, by the triangle inequality, we have
\begin{align}\label{eq:Eb_norm_decomposition}
    \left\Vert \bE_\cB - \gamma \bI_n \right\Vert_2 & \leq \left\Vert \overline{\bE}_\cB \right\Vert_2 + \left\Vert \bE_\cB - \overline{\bE}_\cB - \gamma \bI_n \right\Vert_2 \nonumber \\
    & = \left\Vert \overline{\bE}_\cB \right\Vert_2 + \max_{i} \left\vert \bE_\cB (i,i) - \gamma \right\vert.
\end{align}
It suffices to substitute the probability bounds~\eqref{eq:Eb_diag_chernoff} and~\eqref{eq:bE_offdiag_prob} into~\eqref{eq:Eb_norm_decomposition} with the union bound to yield
\begin{equation}\label{eq:bE_prob}
	   \norm{\bE_\cB - \bI_n}_2 \leq 6\sqrt{2 n} \log(n)  d \max\left\{\frac{1}{m^2}, \frac{n}{m^4}\right\} {(1+\varepsilon)}^2
\end{equation}
with probability greater than $1- 4 n^{-1}$. The statement of Lemma~\ref{lem:concentrationEB} follows by selecting the incoherence parameter $\mu = \left(6\sqrt{2}\left(1+\varepsilon\right)^2 \right)^{\frac{1}{8}}$.\qed

\subsection{Proof of Lemma~\ref{lem:concentrationCovariance}}\label{sec:proof_lemma_concentrationCovariance}

We seek to upper bound the quantity $\left\Vert \hat{\bSigma}_\bz - \bSigma_\bz \right\Vert_2$ with overwhelming probability. We start the proof by recalling in Proposition \ref{prop:matrixBernstein-Covariance} the matrix Bernstein concentration inequality in the case of covariance estimation~\cite{tropp2015introduction}.
\begin{proposition}[Matrix Bernstein for covariance estimation]\label{prop:matrixBernstein-Covariance}
	Assume that there exist a constant $C$ such that $\norm{ \bz_k - \rE[\bz_k]}_2 \leq C \log(L) \norm{\bSigma_\bz}_2 $ for all $k = 1, \dots, L$, we have that
	\begin{multline}\label{eq:matrixBernstein}
		\bbP\left( \left\Vert \hat{\bSigma}_\bz - \bSigma_\bz \right\Vert_2 \geq t \right) \\
		\leq 2n \exp \left(\frac{-L t^2/2}{C \log(L) \left(\left\Vert \bSigma_\bz \right\Vert_2^2 + \frac{2}{3} \norm{\bSigma_\bz}_2 t \right)} \right).
	\end{multline}
\end{proposition}
Hence, providing a high probability bound on $\norm{\bSigma}_2$ is sufficient to prove the desired statement. To that end, we apply the triangle inequality on \eqref{eq:covariance}. This yields
\begin{align}\label{eq:sigma-bound1}
	\MoveEqLeft \norm{\bSigma_\bz}_2 = \norm{p(1-p) \bP \bB \bP + 2p\bE_\cB + \bC}_2 & \nonumber\\
	&\leq p(1-p)  \norm{\bP \bB \bP}_2 + 2p \norm{\bE_\cB - \gamma \bI_n}_2 \nonumber \\
	& \qquad + \norm{\bC}_2 + \norm{2p\gamma \bI_n}_2 \nonumber\\
&  \leq p \norm{\bP}^2_2 \norm{\bB}_2 + 2p \norm{\bE_\cB - \gamma \bI_n}_2  + \norm{\bC}_2 + 2p\gamma.
\end{align}
Now, we individually bound each element on the right-hand side of \eqref{eq:sigma-bound1}. We recall $\norm{\bB}_2 = d$, $\norm{\bE_\cB - \gamma \bI_n}_2$ is controlled by the $(\mu,\nu)$-coherence assumption on $\bA$. Furthermore, we recall that for any Hermitian matrices $\bX, \bY$ of same dimension, we have $\norm{\bX \odot \bY}_2 \leq \norm{\bX}_{2} \norm{\bY}_{\max}$ (see \emph{e.g.}~\cite[113]{johnson1990matrix}). This implies that
\begin{subequations}\label{eq:p-norm-bound1}
\begin{align}
    \norm{\bP}_2 &= \norm{\bM \odot \bM}_2 \leq \norm{\bM}_2 \norm{\bM}_{\max} \nonumber \\
    &= \norm{\bA}_2^2 \norm{\bM}_{\max} \\
	\norm{\bM^2 \odot \bM}_2& \leq \norm{\bM^2}_2 \norm{\bM}_{\max} \nonumber \\
	& = \norm{\bA}_{2}^4 \norm{\bM}_{\max} \\
    \norm{\bA \odot \bA}_2 & \leq \norm{\bA}_2 \norm{\bA}_{\max}.
\end{align}
\end{subequations}
We are now ready to bound $\norm{\bC}_2$ to derive an upper bound on $\norm{\bSigma_\bz}_2$. Applying the triangle inequality on the expression of $\bC$ given in \eqref{eq:C-expression} gives
\begin{align}\label{eq:C-bound1}
		\norm{\bC}_2 & \leq 2 p \norm{\bP}_2^2 + 2 p^2 \norm{\bF}_2 + 2 \sigma^4  \norm{\bG}_2 \nonumber \\
	 & \qquad + 2 p \sigma^2\norm{\bm{M}^2 \odot \bM}_2 + 2 \sigma^4 \norm{\bA \odot \bA}_2^2
\end{align}
Substituting~\eqref{eq:p-norm-bound1} into~\eqref{eq:C-bound1} and leveraging the $(\mu,\nu)$-coherence assumption on the matrix $\bA$ yield
\begin{align}\label{eq:C-bound2}
    \norm{\bC}_2 & \leq 2p \norm{\bA}_2^4 \norm{\bM}_{\max}^2 + 2 p^2 \norm{\bF}_2  + 2\sigma^4 \norm{\bG}_2  \nonumber \\
  & \qquad  + 2 p \sigma^2 \norm{\bA}_2^4 \norm{\bM}_{\max} + 2 \sigma^4 \norm{\bA}_2^2 \norm{\bA}_{\max}^2   \nonumber \\
  & \leq 4p \frac{n^2}{m^2}\log^2(n) \mu^8 + 2p \frac{n^2}{m^2} \log(n)\mu^6 \sigma^2 \nonumber \\
  & \qquad + 4\frac{n^2}{m^2} \log(n) \mu^4 \sigma^4.
\end{align}

Finally, we can substitute~\eqref{eq:C-bound2} into~\eqref{eq:sigma-bound1} to obtain
\begin{align}
\norm{\bSigma_\bz}_2 &\leq 7pd\frac{n^2}{m^2} \log^2(n) \mu^8 + 2p \frac{n^2}{m^2}\log(n)\mu^6\sigma^2  \nonumber \\
& \qquad + 4\frac{n}{m}\log(n)\mu^4\sigma^4 \nonumber \\
&\leq 7pd\frac{n^2}{m^2} \log^2(n) \mu^8 \nonumber \\
& \qquad \cdot \left( 1 + \frac{2}{7\log(n)} \frac{\sigma^2}{\mu^2} + \frac{4}{7 p \log(n)}\frac{\sigma^4}{\mu^4} \right)
\end{align}
We achieve the desired statement with  $p = \frac{m}{\beta n}$ and by letting $t = \frac{ \log(L)}{7\sqrt{L}} \norm{\bSigma_\bz}_2$ in the matrix Bernstein bound~\eqref{eq:matrixBernstein}. \qed.

\renewcommand*{\bibfont}{\footnotesize}
\printbibliography

\end{document}